\documentclass[10pt,a4paper,fleqn]{article}

\usepackage[T1]{fontenc}
\usepackage[utf8]{inputenc}
\usepackage[margin=1in]{geometry}
\usepackage{newtxtext}
\usepackage{mathtools,amsthm,amssymb,array,multirow}
\usepackage{newtxmath}
\usepackage{graphicx,booktabs,etoolbox}
\usepackage{xcolor}
\usepackage{microtype}
\usepackage{threeparttable}
\usepackage{capt-of}
\usepackage{longtable}
\usepackage{tabularx}
\usepackage{enumitem}
\usepackage{needspace}
\usepackage{float}
\usepackage{placeins}
\usepackage[authoryear]{natbib}
\usepackage{hyperref}
\usepackage{bookmark}
\allowdisplaybreaks
\hypersetup{
  hidelinks,
  hypertexnames=false,
  pdftitle={Capacity Modification for Matching with Couples: An Integer Programming and Computational Study},
  pdfauthor={Zexiao Wu; Genjiu Xu; P{\'e}ter Bir{\'o}}
}

\newtheorem{theorem}{Theorem}[section]

\newtheorem{proposition}[theorem]{Proposition}

\theoremstyle{definition}
\newtheorem{definition}{Definition}[section]

\newtheorem{counterexample}{Counterexample}[section]

\newcommand{\D}{\mathcal D}
\newcommand{\Sg}{\mathcal S}
\newcommand{\C}{\mathcal C}
\newcommand{\Hh}{\mathcal H}
\newcommand{\E}{\mathcal E}
\newcommand{\Pk}{\mathcal P}
\newcommand{\Kcross}{\mathcal K}
\newcommand{\outside}{\varnothing}
\newcommand{\ICMNF}{\mbox{ICM-NF}}
\newcommand{\ICMBC}{\mbox{ICM-BC}}
\newcommand{\ICMUC}{\mbox{ICM-UC}}
\newcommand{\IFCCUT}{\mbox{IFC-CUT}}
\newcommand{\ICMCUTNF}{\mbox{ICM-CUT-NF}}
\newcommand{\ICMCUTBC}{\mbox{ICM-CUT-BC}}
\newcommand{\ICMCUTUC}{\mbox{ICM-CUT-UC}}
\newcommand{\NV}{\mbox{NV}}

\title{Capacity Modification for Matching with Couples:\\
An Integer Programming and Computational Study}
\author{%
  Zexiao Wu\textsuperscript{1,2}\qquad
  Genjiu Xu\textsuperscript{1,2}\qquad
  P{\'e}ter Bir{\'o}\textsuperscript{3,4}\thanks{Corresponding author.}\\[0.7em]
  \small\textsuperscript{1}\begin{tabular}[t]{@{}c@{}}College of Mathematics and Statistics, Northwestern\\
  Polytechnical University, Xi'an, China\end{tabular}\\
  \small\textsuperscript{2}\begin{tabular}[t]{@{}c@{}}Shenzhen Research Institute, Northwestern Polytechnical University,\\
  Shenzhen, China\end{tabular}\\
  \small\textsuperscript{3}Corvinus University of Budapest, Budapest, Hungary\\
  \small\textsuperscript{4}Institute of Economics, ELTE KRTK, Budapest, Hungary\\[0.6em]
  \small\texttt{wuzexiao@mail.nwpu.edu.cn}\qquad
  \small\texttt{xugenjiu@nwpu.edu.cn}\qquad
  \small\texttt{biro.peter@krtk.mta.hu}%
}
\date{}

\begin{document}
\maketitle

\begin{abstract}
We study capacity modification in the Hospitals/Residents problem with Couples, in which hospital capacities and a stable matching are jointly determined. 
Exact and Cutoff integer programming formulations are developed for the resulting problem.
We assume that, for each couple, either both members are matched or neither is.
For arbitrary strict preferences, every instance admits a stable matching after changing each hospital capacity by at most two.
Total capacity increases by at most one, and this bound is tight.
Under sub-responsive and sub-complete couple preferences, a stable matching
exists after adding at most one position to each hospital. The total number of
additions is at most the number of couples, and this bound is also tight.
Computational experiments show that Exact formulations are modestly faster at
fixed capacities, whereas Cutoff formulations perform substantially better in
the minimum repair experiments and attain
higher completion rates in markets with up to 30,000 doctors.
Across 96,000 synthetic and reconstructed markets, small capacity changes restore existence whenever no stable matching exists at the original capacities.
\end{abstract}

\begin{center}
\end{center}

\section{Introduction}
\label{sec:introduction}

Most matching models treat capacities as exogenous, even though market
organizers often determine them before or alongside the assignment process.
This disconnect is especially consequential in residency matching and other
centralized markets, where capacity decisions can determine whether a stable
matching exists. 
Stability is a central design objective in these markets. It requires that no
resident or couple can form a blocking coalition with one or more hospitals.
In the classical Hospitals/Residents problem,
deferred acceptance always yields a stable matching \citep{gale1962}.
Couples invalidate this guarantee because they rank pairs of positions
jointly.
These complementarities are important in medical matching
\citep{roth1984,rothperanson1999}, but they also create a
fundamental difficulty: a stable matching may not exist, and deciding whether
one exists is NP-complete \citep{ronn1990}.

Capacity modification provides one way to restore existence.
\citet{nguyen2018} show that every market with couples admits a stable
matching after changing each hospital capacity by at most two, while
increasing total capacity by at most four.
They also leave open whether the bound for each hospital
can be reduced from two to one. Under sub-responsive and sub-complete couple
preferences, \citet{csaji2026} obtain this unit bound. Their stated guarantee
permits both increases and reductions and does not provide an aggregate
capacity bound. These results leave distinct questions when capacities may
only increase or when total capacity must be preserved.

We study the Hospitals/Residents problem with Couples (HRC) under the
stability concept of \citet{biro2011}, termed BIS-stability by
\citet{delorme2021}. We distinguish three capacity policies. Under near
feasible capacity modification, each hospital may gain or lose at most two
positions while total capacity increases by at most one. Under balanced
capacity modification, each hospital may gain or lose at most one position
while total capacity is preserved. Under unit capacity expansion, each
hospital may gain at most one position and the total number of additions is
limited by a budget. For all three policies, we ask whether stability can be
restored and what the smallest modification is.

Our first contribution sharpens the near-feasibility guarantee of \citet{nguyen2018}. We assume that the two members of a couple are either jointly assigned to an acceptable hospital pair or both unmatched. Under this assumption, the result holds for arbitrary strict preferences. Each hospital capacity changes by at most two, as in their result, while total capacity increases by at most one rather than four. This bound is tight. A stable matching need not exist if total capacity must remain unchanged, even when each hospital capacity may change by at most two. The proof modifies their iterative-rounding argument at the step where the aggregate capacity constraint is discarded and uses the fact that every acceptable couple assignment occupies exactly two hospital positions. Capacity increases at some hospitals may be offset by reductions at others. We also provide a simple two-hospital counterexample showing that balanced unit modification does not guarantee the existence of a stable matching, even with a common master list and sub-responsive and sub-complete couple preferences.

Our second contribution is an existence guarantee for unit capacity expansion. 
Let \(\C\) denote the set of couples. Under sub-responsive and sub-complete
couple preferences, stability can always be restored by adding at most one
position to each hospital and at most \(|\C|\) positions in total.
The aggregate bound is tight even when all hospitals share a common master
list and every acceptable couple assignment uses distinct hospitals. By sharpening the analysis of the
Stable Fixtures rounding of \citet{csaji2026}, the stable matching and the required expansions can be found in
time linear in the compact input size.

Our third contribution is an integer programming framework for jointly
determining a stable matching and hospital capacities under the three policies.
Starting from the fixed capacity formulation of \citet{delorme2021}, we
develop integer programs for deciding existence and minimizing the required
capacity modification under each policy.
We extend the bounded capacity linearization used by \citet{bobbio2026} for
markets without couples to all three types of blocking coalition in the
stability definition. We also adapt the cutoff approach of
\citet{agoston2022} to allow capacities to be decision variables. We prove
that the Exact and Cutoff formulations yield exactly the same stable matchings
and capacity vectors, both for fixed capacities and under each modification
policy.
Table~\ref{tab:capacity-theory-map} summarizes the resulting existence
boundary. Write
\(\delta_h=q'_h-q_h\) and
\(\Delta(q')=\sum_{h\in\Hh}\delta_h\).

\begin{table}[H]
\centering
\begin{threeparttable}
\caption{Existence boundary under capacity modification}
\label{tab:capacity-theory-map}
\footnotesize
\setlength{\tabcolsep}{2.0pt}
\renewcommand{\arraystretch}{1.12}
\begin{tabular*}{\textwidth}{@{\extracolsep{\fill}}ccll@{}}
\toprule
\multicolumn{2}{c}{Capacity limits}
& \multicolumn{2}{c}{Existence boundary} \\
\cmidrule(lr){1-2}\cmidrule(lr){3-4}
\(\delta_h\)
& \(\Delta(q')\)
& Preference domain
& Result \\
\midrule
\multirow[c]{3}{*}{\([-2,2]\)}
& \(0\)
& General HRC
& \textbf{Counterexample} [Proposition~\ref{prop:a2b0-counterexample} (BIS, KPR, MM, \NV{})] \\
& \([0,1]\)
& General HRC
& \textbf{Guaranteed} [Theorem~\ref{thm:general-capacity-bound} (BIS); tight by Proposition~\ref{prop:a2b0-counterexample}] \\
& \([0,4]\)
& General HRC
& \textbf{Guaranteed} [\citet[Theorem~1]{nguyen2018} (BIS, \NV{})] \\
\cmidrule(l){2-4}
\multirow[c]{3}{*}{\([-1,1]\)}
& \(0\)
& General HRC
& \textbf{Counterexample} [Counterexample~\ref{ex:ml-responsive-bc-counterexample} (BIS, KPR, MM, \NV{})] \\
& Unrestricted
& Sub-responsive and sub-complete
& \textbf{Guaranteed} [\citet[Theorem~3.1]{csaji2026} (MM)] \\
& Unrestricted
& General HRC
& \textbf{Open} \\
\cmidrule(l){2-4}
\multirow[c]{2}{*}{\([0,1]\)}
& \([0,|\C|]\)
& Sub-responsive and sub-complete
& \textbf{Guaranteed} [Theorem~\ref{thm:one-sided-near-feasibility} (BIS); tight by Proposition~\ref{prop:tight-one-sided-budget}] \\
& Unrestricted
& General HRC
& \textbf{Open} \\
\bottomrule
\end{tabular*}
\end{threeparttable}
\end{table}

We complement these theoretical contributions with a computational study of
formulation performance, scalability, and capacity modification. Our
implementation reproduces all 360 archived outcomes reported by
\citet{delorme2021}. Across 510 variants with strict preferences and fixed
capacities, Exact and Cutoff return the same existence and maximum cardinality
results, with Exact modestly faster. For minimum repair under balanced
capacity modification and unit capacity expansion, Cutoff performs better
under all four stability notions.
Its advantage is largest under the MM and \NV{} notions of stability, where
mean running times fall by factors of 41 to 55.
A separate scalability experiment considers 500 independently generated
markets with 100 to 30,000 doctors. Under a 600-second limit shared by the
fixed capacity and minimum repair stages, Cutoff completes 93.4\% of the
balanced modification procedures and 95.4\% of the unit expansion procedures.
The corresponding rates for Exact are 71.4\% and 72.2\%.

The capacity experiments consider 84,000 markets with \(n=100\) and 12,000
reconstructed markets with \(n=500\). At the original capacities, 10,640
markets in the first sample and 2,190 markets in the second sample do not admit
a stable matching. Unit expansion repairs all 12,830 markets, and one added
position is sufficient in 96.4\% of them. Balanced modification repairs all but
19 markets, all of which belong to the \(n=100\) sample. In these markets,
allowing each hospital to gain or lose up to two positions restores stability
in all 10,640 markets without changing total capacity. Among the 60,000 markets
whose total capacity equals the number of doctors, all 6,087 markets without a
stable matching at the original capacities are repaired by a single balanced
transfer. For markets that are stable at the original capacities, unit
expansion also increases maximum matching cardinality, although the improvement
becomes smaller as the budget grows. These results do not provide a general
guarantee, but they suggest first considering a small capacity reallocation
when total capacity is sufficient and adding positions only when necessary.

The remainder of the paper is organized as follows.
Section~\ref{sec:related-work} reviews the related literature, and Section~\ref{sec:model} introduces the HRC model.
Section~\ref{sec:main-results} presents the structural results, and
Section~\ref{sec:ip} develops the integer programming models.
Section~\ref{sec:model-improvements} gives the cutoff formulations,
and Section~\ref{sec:computation} reports the computational
study.

\section{Related work}
\label{sec:related-work}

The literature on HRC with fixed capacities identifies both the limits of
general preferences and conditions under which stability can be recovered.
\citet{mcdermid2010} show that the existence problem remains hard even under
severe restrictions on preference list lengths and hospital capacities.
Stable matchings are guaranteed under restricted preference domains such as
weak responsiveness, and constructive procedures are available in this
setting \citep{klausklijn2005,klausklijn2007}. In large random markets, stable
matchings exist with high probability when couples are sufficiently scarce,
whereas nonexistence may persist when couples form a fixed share of the market
\citep{kojima2013,ashlagi2014}. \citet{biroklijn2013} provide a broader survey
of the economic and algorithmic literature.

The computational literature is closely connected to the operation of
centralized medical matching. The redesign of the National Resident Matching
Program combines market design theory, practical implementation, and
computational experiments \citep{rothperanson1999}. Subsequent work develops
heuristics, methods based on Scarf's lemma, and exact integer programming
formulations for HRC
\citep{biro2011,biro2016,biro2014ip,delorme2021}. These approaches treat
hospital capacities as fixed.
Integer programming has also been applied to other matching and
allocation problems
\citep{delorme2019,birogudmundsson2021,agoston2022,birogyetvai2023}.
Our integer programming formulation builds directly on
\citet{delorme2021}, while our main models make capacities part of the
optimization.

The closest work uses near feasibility to recover stability in HRC.
\citet{nguyen2018} show that every instance admits a stable matching after
changing each hospital's capacity by at most two, while aggregate capacity is
never reduced and increases by at most four. Motivated by the open question of \citet{nguyen2018} on unit
individual changes, \citet{csaji2026} reduce the individual bound to one for
sub-responsive and sub-complete couple preferences. Their stated guarantee
uses the MM-stability, permits both additions and removals, and
does not provide an aggregate capacity bound. \citet{cembrano2025} recover the
Nguyen--Vohra guarantee in a multidimensional allocation framework and combine
stability with approximate group fairness. Related rounding methods address
proportionality, multidimensional constraints, and other complex matching
structures \citep{nguyen2019,nguyen2021,csaji2025}.

Capacity expansion has also been studied in many-to-one markets without
couples. One strand asks how a limited pool of additional seats should be
allocated to improve assignments or obtain outcomes that respect priorities
\citep{afacan2024,bobbio2026}. Methodologically, \citet{bobbio2026} are closest
to our work: they jointly optimize school expansions and a student-optimal
stable assignment using exact formulations, cutting planes, and heuristics. A
second strand minimizes the augmentation needed to obtain complete or
efficient stable matchings, incorporates heterogeneous seat costs, or restores
strong stability when hospital preferences contain ties
\citep{chenCsaji2023,balasundaram2026,ranjan2025}.

A broader line studies capacity reductions as well as expansions and asks
which outcomes can be implemented or sustained through capacity changes.
\citet{gokhale2024} examine capacity changes that ensure a designated worker
and firm are matched or make a prescribed matching stable. They also compare
capacity modification with preference manipulation by firms.
\citet{bobbiovariation2026} analyze resident welfare under separate expansion
and reduction budgets. In Stable Fixtures, a
nonbipartite model in which agents may have multiple partners,
\citet{glitzner2026} studies minimal capacity changes together with deviation
incentives relative to the original capacities. Most of the many-to-one
studies start from markets in which a stable matching exists under the
ordinary definition at the original capacities. In HRC, ordinary stability itself
may fail, making restoration of existence the first capacity planning question.

\section{Matching markets with couples}
\label{sec:model}

We consider the Hospitals/Residents problem with Couples (HRC).
Let \(\D\) be the finite set of doctors and \(\Hh\) the finite set of hospitals.
The doctors are partitioned into a set \(\Sg\) of single doctors and a set
$\C=\{(d_k^1,d_k^2):k=1,\ldots,|\C|\}$
of couples. The two members of a couple are distinct, no doctor belongs to
more than one couple, and a doctor who belongs to a couple is called linked.
Each hospital \(h\in\Hh\) has an integer capacity \(q_h\geq0\).
Let \(q=(q_h)_{h\in\Hh}\) denote the
capacity vector. Throughout, a capacity vector has nonnegative integer
components.

Each doctor \(i\in\D\) has a strict preference order \(\succ_i\) over a set
\(H(i)\subseteq\Hh\) of acceptable hospitals. Every hospital in \(H(i)\) is
preferred to the unmatched outcome \(\outside\); hospitals outside \(H(i)\)
are unacceptable. Let
$\E=\{(i,h):i\in\D,\ h\in H(i)\}$
be the set of acceptable doctor--hospital pairs. Each hospital \(h\) has a
strict priority order \(\succ_h\) over the doctors who regard \(h\) as
acceptable. 

For each couple \(k\in\C\), let
\(\Pk_k\subseteq H(d_k^1)\times H(d_k^2)\)
be the set of acceptable ordered hospital pairs.
The couple has a strict preference order \(\succ_k\) over
\(\Pk_k\cup\{(\outside,\outside)\}\), and every pair in \(\Pk_k\) is preferred
to \((\outside,\outside)\). We assume that the two members of a couple are
either jointly assigned to an acceptable hospital pair or both unmatched.
This assumption captures
complementarity in joint employment and relocation decisions, under which the
value of one member's placement depends on the other member also receiving an
acceptable position. An assignment that places only one member is therefore
unacceptable to the couple.

A matching \(M\) assigns each single doctor \(i\in\Sg\) either to a hospital
in \(H(i)\) or to \(\outside\). It assigns each couple \(k=(a,b)\) either to
one pair in \(\Pk_k\) or to \((\outside,\outside)\). For a hospital \(h\), let
$M(h)=\{i\in\D:M(i)=h\}$.
The matching is \emph{feasible} under capacity vector \(q\) if
\(|M(h)|\leq q_h\) for every \(h\in\Hh\). Hospital \(h\) is full when
\(|M(h)|=q_h\) and undersubscribed when \(|M(h)|<q_h\).

We use the stability definition of \citet{biro2011}, called BIS-stability by
\citet{delorme2021}. It distinguishes deviations by a single doctor, a couple
to two distinct hospitals, and a couple to two positions at the same hospital.

\begin{definition}
\label{def:bis-blocking}
Let $\mathcal{I}$ be an instance of HRC and \(M\) be a feasible matching of $\mathcal{I}$.

\begin{enumerate}[label=\textnormal{(S\arabic*)},leftmargin=3.1em]
  \item A single doctor \(i\) and hospital \(h\in H(i)\) block \(M\) if
  \(i\) is unmatched or \(h\succ_i M(i)\), and either \(h\) is
  undersubscribed or \(i\succ_h j\) for some \(j\in M(h)\).

  \item A couple \(k=(a,b)\) and two distinct hospitals \(h\neq g\) block
  \(M\) if the couple is unmatched or
  \((h,g)\succ_k(M(a),M(b))\), and both proposed assignments can be
  accommodated. Hospital \(h\) can accommodate \(a\) if \(M(a)=h\), if
  \(h\) is undersubscribed, or if \(a\succ_h j\) for some \(j\in M(h)\).
  Hospital \(g\) can accommodate \(b\) under the analogous condition.

  \item A couple \(k=(a,b)\) and one hospital \(h\), with
  \((h,h)\in\Pk_k\), block \(M\) if the couple is unmatched or
  \((h,h)\succ_k(M(a),M(b))\), and one of the following conditions holds:
  \begin{enumerate}[label=\textnormal{(\roman*)},leftmargin=2.5em]
    \item \(h\) has at least two free positions;
    \item \(h\) has exactly one free position and either
    \(h\in\{M(a),M(b)\}\), or hospital \(h\) prefers both \(a\) and \(b\)
    to some member of \(M(h)\);
    \item \(h\) is full and at least one of the following applies:
    \begin{enumerate}[label=\textnormal{(\alph*)},leftmargin=2.5em]
      \item \(h\in\{M(a),M(b)\}\), and hospital \(h\) prefers both \(a\)
      and \(b\) to some member of \(M(h)\);
      \item hospital \(h\) prefers both \(a\) and \(b\) to a linked doctor
      \(j\in M(h)\) whose partner is also assigned to \(h\);
      \item hospital \(h\) prefers both \(a\) and \(b\) to at least two distinct
      members of \(M(h)\).
    \end{enumerate}
  \end{enumerate}
\end{enumerate}
A feasible matching is \emph{stable} if it admits no blocking coalition.
\end{definition}

The way S3 treats deviations in which both members of a couple apply to the
same hospital distinguishes BIS stability from other notions of stability
for HRC \citep{delorme2021}.
For comparison,
Appendix~\ref{app:alternative-stability} summarizes the KPR and MM conditions
and records the stronger condition satisfied by the Nguyen--Vohra
construction (\NV{}).
Unless stated otherwise, a matching described as stable is BIS-stable
throughout this paper.  
 
 \section{Existence results}
 \label{sec:main-results}
 
 The following result of \citet[Theorem~1]{nguyen2018} provides a benchmark
 for our aggregate result and for the more restrictive forms of capacity
 modification studied below.
 
 \begin{theorem}[Nguyen and Vohra, \citeyear{nguyen2018}]
 	\label{thm:nguyen-vohra-capacity-bound}
 	Suppose that single doctors and couples have strict preferences over their
 	feasible outcomes and that hospital preferences are responsive to strict
 	priority orders over doctors. For any initial capacity vector \(q\), there
 	exist a capacity vector \(q'\) and a matching that is stable under \(q'\) such
 	that
 	\begin{equation*}
 		|q'_h-q_h|\leq2\quad\forall h\in\Hh,
 		\qquad
 		0\leq\sum_{h\in\Hh}(q'_h-q_h)\leq4.
 	\end{equation*}
 \end{theorem}
 
 The Scarf construction of \citet{nguyen2018} uses a technical order over the
 assignment columns incident to each hospital. In the proof of
 Theorem~\ref{thm:general-capacity-bound}, we state this order explicitly and
 verify directly that domination excludes each of the three BIS blocking
 types. Thus their capacity guarantee, when specialized to our model, also
 yields a BIS-stable matching.
 
 \subsection{A tight near feasibility guarantee}
 
When the two members of every couple must either be assigned to an acceptable hospital pair or both remain unmatched, a stable matching can be obtained by changing each hospital’s capacity by at most two and increasing total capacity by at most one.
 
 \begin{theorem}
 	\label{thm:general-capacity-bound}
 	Every HRC instance
 	admits a stable matching under a capacity vector \(q'\) satisfying
 	\begin{equation*}
 		|q'_h-q_h|\leq2\quad\forall h\in\Hh,
 		\qquad
 		0\leq\sum_{h\in\Hh}(q'_h-q_h)\leq1.
 	\end{equation*}
 \end{theorem}
 
 The proof is given in Appendix~\ref{app:general-capacity-bound-proof}. It
 retains the Scarf domination and hospital-row rounding of
 \citet{nguyen2018}, but replaces their final deletion of the aggregate
 capacity row by a grouped constraint on the at most two remaining slack
 applicant rows. The argument is existential and does not assert a
 polynomial-time algorithm. Proposition~\ref{prop:a2b0-counterexample} shows
 that the upper bound one is tight.
 
  \begin{proposition}
 	\label{prop:a2b0-counterexample}
 	There exists an HRC instance with three hospitals, four
 	couples, and two single doctors that admits no stable matching under any
 	capacity vector \(q'\) satisfying
 	\[
 	|q'_h-q_h|\leq2\quad\forall h\in\Hh,
 	\qquad
 	\sum_{h\in\Hh}q'_h=\sum_{h\in\Hh}q_h.
 	\]
 \end{proposition}
 
 The counterexample is given in
 Appendix~\ref{app:a2b0-counterexample}. 
 Every acceptable couple assignment in the instance involves two distinct hospitals, so case S3 does not arise and all stability notions considered in this paper coincide for this instance.
 
 \subsection{Balanced capacity modification}
 \label{subsec:deterministic-theory}
 
 The result in the previous subsection applies to arbitrary strict preferences.
 The theorem below assumes that couple preferences are sub-responsive and
 sub-complete. A couple is sub-responsive if there are strict individual
 preference orders for its two members such that, whenever two acceptable
 hospital pairs differ only in one member's assignment, the couple's ranking
 of those pairs agrees with that member's order. Sub-completeness requires every
 pair of individually acceptable hospitals to be acceptable to the couple. Thus,
 for a couple \(k=(a,b)\), \(\Pk_k=H(a)\times H(b)\).
 
 \begin{theorem}[Cs\'aji et al., \citeyear{csaji2026}]
 	\label{thm:csaji-unit-capacity-bound}
 	Suppose that every couple has sub-responsive and sub-complete preferences.
 	For any initial capacity vector \(q\), there exist a capacity vector \(q'\)
 	and an MM-stable matching under \(q'\) such that
 	\(|q'_h-q_h|\leq1\) for every \(h\in\Hh\).
 \end{theorem}
 
Theorem~\ref{thm:csaji-unit-capacity-bound} allows capacity increases and reductions but does not bound the change in total capacity. The following counterexample shows that a stable matching need not exist when total capacity is unchanged.
 
 \begin{counterexample}
 	\label{ex:ml-responsive-bc-counterexample} 
 	Let \(\Hh=\{h_1,h_2\}\), let \(q=(2,2)\), and consider one couple
 	\(k=(d_1,d_2)\) and three single doctors \(s_1,s_2,s_3\). The individual
 	preference orders are \(d_1\colon h_2\succ h_1\),
 	\(d_2\colon h_1\succ h_2\), \(s_1\colon h_1\succ h_2\), and
 	\(s_2,s_3\colon h_2\succ h_1\). Both hospitals use the common priority order
 	\(d_1\succ s_1\succ s_2\succ s_3\succ d_2\). The couple's sub-responsive and sub-complete ranking of the four hospital pairs is
 	\((h_2,h_1)\succ_k(h_2,h_2)\succ_k(h_1,h_1)\succ_k(h_1,h_2)\).
 	
 	The three capacity vectors permitted by balanced
 	capacity modification are
 	\((1,3)\), \((2,2)\), and \((3,1)\). If the couple is assigned, at least
 	one single doctor is unmatched because total capacity is four. Every single
 	doctor accepts both hospitals and has higher priority than \(d_2\). An
 	unmatched single therefore forms an S1 block with the hospital to which
 	\(d_2\) is assigned, using an unfilled position if one is available and
 	otherwise displacing \(d_2\). Hence a matching without an S1 block must leave
 	the couple unmatched.
 	
 	If the couple is unmatched, all three singles must be assigned; otherwise an
 	unmatched single would block an acceptable hospital with an unfilled
 	position. The assignments below are forced by the singles' preferences and
 	the common priority order.
 
 	Under \((1,3)\), \(s_1\) must be assigned to \(h_1\); otherwise \(s_1\)
 	would block \(h_1\), which would either have an unfilled position or be
 	occupied by a lower priority single. Hence \(s_2,s_3\) are assigned to
 	\(h_2\). The couple then forms an S2 block with \((h_1,h_2)\), because
 	\(d_1\) can displace \(s_1\) at \(h_1\) and \(d_2\) can use the unfilled
 	position at \(h_2\).
 
 	Under \((2,2)\), assigning \(s_2\) to \(h_1\) creates an S1 block at
 	\(h_2\), where \(s_2\) can use an unfilled position or displace \(s_3\).
 	If \(s_3\) is assigned to \(h_1\), then either \(s_3\) blocks \(h_2\), or
 	\(h_2\) is full with \(s_1,s_2\), in which case \(s_1\) blocks the unfilled
 	position at \(h_1\). Thus \(s_1\) is assigned to \(h_1\) and \(s_2,s_3\)
 	are assigned to \(h_2\). The couple forms an S2 block with \((h_2,h_1)\),
 	because \(d_1\) can displace a single doctor at \(h_2\) and \(d_2\) can use
 	the unfilled position at \(h_1\).
 
 	Under \((3,1)\), the position at \(h_2\) must be assigned to \(s_2\). Leaving
 	it unfilled gives \(s_2\) an S1 block. If it is assigned to \(s_3\), then
 	\(s_2\) can displace \(s_3\); if it is assigned to \(s_1\), then \(s_1\)
 	prefers the unfilled position at \(h_1\). Hence \(s_1,s_3\) are assigned to
 	\(h_1\). The couple again forms
 	an S2 block with \((h_2,h_1)\), because \(d_1\) can displace \(s_2\) at
 	\(h_2\) and \(d_2\) can use the unfilled position at \(h_1\). Thus none of
 	the three capacity vectors admits a stable matching.
 \end{counterexample}
 
Because S1 and S2 are identical under all four definitions, the counterexample applies to all four stability notions considered in this paper.
 
 \subsection{Unit capacity expansion}
 \label{subsec:structured-main-results}
 
 Under the same preference restrictions, the next result establishes
 BIS-stability using only capacity increases and bounds the total number of
 positions added. Let \(L\) denote the size of the input used by the algorithm, including
 hospital capacities, single doctors' preferences, hospital priorities,
 and the two individual orders for each couple that establish
 sub-responsiveness.
 
 \begin{theorem}
 	\label{thm:one-sided-near-feasibility}
 	Consider any HRC instance in which couple preferences are sub-complete and
 	sub-responsive, and let \(q\) be its capacity vector. There exist a matching
 	\(M\) and a capacity vector \(\widehat q\) such that \(M\) is stable under
 	\(\widehat q\) and
 	\begin{equation*}
 		0\leq \widehat q_h-q_h\leq1\quad\forall h\in\Hh,
 		\qquad
 		\sum_{h\in\Hh}(\widehat q_h-q_h)\leq|\C|.
 	\end{equation*}
 	The matching and capacities can be found in \(O(L)\) time.
 \end{theorem}
 
 The proof is given in Appendix~\ref{app:one-sided-near-feasibility-proof}.
 For an instance \(\mathcal I\), let \(B_{\min}(\mathcal I)\) denote the
 smallest budget for which unit capacity
 expansion is feasible. The next result identifies
 the exact worst case budget of
 Theorem~\ref{thm:one-sided-near-feasibility}.
 
 \begin{proposition}
 	\label{prop:tight-one-sided-budget}
 	For every integer \(m\geq 1\), there exists an HRC instance
 	\(\mathcal I\) with \(m\) couples and sub-complete and sub-responsive
 	couple preferences such that $B_{\min}(\mathcal I)=m$.
 \end{proposition}
 
 The lower bound construction is given in
 Appendix~\ref{app:tight-one-sided-budget-proof}. Thus \(|\C|\) is not only
 sufficient; it is the smallest budget that guarantees feasibility throughout
 this preference domain.

\section{Integer programming formulations}
\label{sec:ip}

This section presents integer programming formulations for stable matching
at fixed capacities and with capacity modification.
Section~\ref{subsec:fixed-ip} gives the formulation for the fixed capacity
vector \(q\). Section~\ref{subsec:variable-ip} makes capacities decision
variables and derives linear stability constraints for bounded capacity
changes. Section~\ref{subsec:special-capacity-modifications} applies this
formulation to the three capacity policies.

For each \((i,h)\in\E\), let \(x_{ih}=1\) if doctor \(i\) is assigned to
hospital \(h\), and let \(x_{ih}=0\) otherwise. For each couple \(k\in\C\)
and pair \(p\in\Pk_k\), let \(y_{kp}=1\) if the couple is assigned to \(p\).
The basic feasibility constraints are
\begin{align}
  &\sum_{h\in H(i)}x_{ih}\leq1
  &&\forall i\in\D,                                      \label{eq:doctor}\\
  &\sum_{p\in\Pk_k}y_{kp}\leq1
  &&\forall k\in\C,                                      \label{eq:couple}\\
  &x_{d_k^1h}
  =\sum_{\substack{p=(h,g)\in\Pk_k}}y_{kp}
  &&\forall k\in\C,\ h\in H(d_k^1),                      \label{eq:link1}\\
  &x_{d_k^2h}
  =\sum_{\substack{p=(g,h)\in\Pk_k}}y_{kp}
  &&\forall k\in\C,\ h\in H(d_k^2).                      \label{eq:link2}
\end{align}
Constraints~\eqref{eq:link1}--\eqref{eq:link2} connect the individual and
couple variables. If \(y_{k,(h,g)}=1\), the first member of couple \(k\) is
assigned to \(h\) and the second member is assigned to \(g\). If
\(y_{kp}=0\) for every \(p\in\Pk_k\), neither member is assigned.

\subsection{IP with fixed capacities (IFC)}
\label{subsec:fixed-ip}

Under the capacity vector \(q\), the hospital feasibility constraint is
\begin{equation}
  \sum_{i:(i,h)\in\E}x_{ih}\leq q_h
  \qquad \forall h\in\Hh.
  \label{eq:fixed-capacity}
\end{equation}
Let $H_h^{\succeq}(i):=\{g\in H(i):g\succeq_i h\}$ and  $\D_i^{\succeq}(h):=\{j\in\D:(j,h)\in\E,\ j\succeq_h i\}$. The following inequalities exclude S1 blocks:
\begin{equation}
  q_h\left(1-
  \sum_{g\in H_h^{\succeq}(i)}x_{ig}\right)
  \leq
  \sum_{j\in\D_i^{\succeq}(h)}x_{jh}
  \qquad \forall i\in\Sg,\ h\in H(i).
  \label{eq:fixed-S1}
\end{equation}
If doctor \(i\) is not assigned to \(h\) or to a hospital that \(i\) prefers
to \(h\), then \(h\) must be full with doctors to whom it gives at least as
high a priority as to \(i\).

Let $\Pk_p^{\succeq}(k):=\{p'\in\Pk_k:p'\succeq_k p\}$. For every couple \(k=(d_k^1,d_k^2)\) and option
\(p=(h,g)\in\Pk_k\) with \(h\neq g\), S2 is excluded by
\begin{align}
  q_h\left(1-
  \sum_{p'\in\Pk_p^{\succeq}(k)}y_{kp'}
  -\alpha_{kp}^{1}\right)
  &\leq
  \sum_{j\in\D_{d_k^1}^{\succeq}(h)}x_{jh}-x_{d_k^1h},
  &&\forall k\in\C,\ p=(h,g)\in\Pk_k,\ h\neq g,
  \label{eq:fixed-S2a}\\
  q_g\left(1-
  \sum_{p'\in\Pk_p^{\succeq}(k)}y_{kp'}
  -\alpha_{kp}^{2}\right)
  &\leq
  \sum_{j\in\D_{d_k^2}^{\succeq}(g)}x_{jg}-x_{d_k^2g},
  &&\forall k\in\C,\ p=(h,g)\in\Pk_k,\ h\neq g,
  \label{eq:fixed-S2b}\\
  \alpha_{kp}^{1}+\alpha_{kp}^{2}&\leq1,
  &&\forall k\in\C,\ p=(h,g)\in\Pk_k,\ h\neq g.
  \label{eq:fixed-S2c}
\end{align}
When couple \(k\) prefers \(p\) to its current assignment,
constraint~\eqref{eq:fixed-S2c} ensures that at least one of
\eqref{eq:fixed-S2a} and~\eqref{eq:fixed-S2b} remains active. The active
constraint requires the corresponding hospital to be full with doctors it
ranks at least as highly as the proposed member.
It remains to exclude S3. Consider \(p=(h,h)\in\Pk_k\), and relabel the
members of \(k\) as \(i_1,i_2\) so that \(i_1\succ_h i_2\). Define
\begin{equation*}
  	\Kcross_h(k,i_2):=
  	\left\{
  	\ell\in\C\setminus\{k\}:
  	(h,h)\in\Pk_\ell,\
  	d_\ell^1\prec_h i_2\prec_h d_\ell^2
  	\ \text{or}\
  	d_\ell^2\prec_h i_2\prec_h d_\ell^1
  	\right\}.
\end{equation*}
Therefore, S3 is excluded by
\begin{equation}
  \begin{aligned}
  q_h\left(1-
  \sum_{p'\in\Pk_p^{\succeq}(k)}y_{kp'}\right)
  -1+x_{i_1h}+x_{i_2h}
  +\sum_{\ell\in\Kcross_h(k,i_2)}y_{\ell,(h,h)}
  \leq
  \sum_{j\in\D_{i_2}^{\succeq}(h)}x_{jh},
  &&\forall k\in\C,\ p=(h,h)\in\Pk_k,\ i_1\succ_h i_2.
  \end{aligned}
  \label{eq:fixed-S3}
\end{equation}
When both members of a couple apply to the same hospital, the BIS constraint
in \citet[Eq.~(18)]{delorme2021} reduces to \eqref{eq:fixed-S3} under strict
preferences.

IFC is given by
\begin{equation}
  \begin{aligned}
  \max\quad &\sum_{i\in\D}\sum_{h\in H(i)}x_{ih}\\
  \text{subject to}\quad
  &\eqref{eq:doctor}\text{--}\eqref{eq:fixed-S3},\\
  &x_{ih}\in\{0,1\}
  &&\forall (i,h)\in\E,\\
  &y_{kp}\in\{0,1\}
  &&\forall k\in\C,\ p\in\Pk_k,\\
  &\alpha_{kp}^1,\alpha_{kp}^2\in\{0,1\}
  &&\forall k\in\C,\ p=(h,g)\in\Pk_k,\ h\neq g.
  \end{aligned}
  \label{eq:fixed-model}
\end{equation}
This is the specialization of the fixed capacity BIS-formulation in
\citet[Section~4]{delorme2021} to strict preferences and couples that cannot
be partially assigned. Consequently, IFC is feasible if and only if a stable
matching exists under \(q\). When IFC is feasible, its objective selects a
stable matching with the maximum number of matched doctors.

\subsection{IP with capacity modification (ICM)}
\label{subsec:variable-ip}

ICM determines hospital capacities together with a matching that is stable
under those capacities. For each hospital \(h\), let \(u_h^+\) and
\(u_h^-\) be nonnegative integer variables for the numbers of positions
added and removed, respectively. The modified capacity is
\begin{equation}
  q'_h=q_h+u_h^+-u_h^-\geq0
  \qquad\forall h\in\Hh.
  \label{eq:general-capacity-change}
\end{equation}
Hospital feasibility becomes
\begin{equation}
  \sum_{i:(i,h)\in\E}x_{ih}
  \leq q'_h
  \qquad \forall h\in\Hh.
  \label{eq:variable-capacity}
\end{equation}

The stability constraints are obtained by replacing \(q\) with \(q'\) in
\eqref{eq:fixed-S1}--\eqref{eq:fixed-S3}. Thus the general ICM is
\begin{equation}
 \begin{aligned}
 \max\quad &\sum_{i\in\D}\sum_{h\in H(i)}x_{ih}\\
 \text{subject to}\quad
  &\eqref{eq:doctor}\text{--}\eqref{eq:link2},\
  \eqref{eq:general-capacity-change},\
  \eqref{eq:variable-capacity},\\
  &\eqref{eq:fixed-S1}\text{--}\eqref{eq:fixed-S3}
  \quad\text{with \(q\) replaced by \(q'\)},\\
  &x_{ih}\in\{0,1\}
  &&\forall (i,h)\in\E,\\
  &y_{kp}\in\{0,1\}
  &&\forall k\in\C,\ p\in\Pk_k,\\
  &\alpha_{kp}^1,\alpha_{kp}^2\in\{0,1\}
  &&\forall k\in\C,\ p=(h,g)\in\Pk_k,\ h\neq g,\\
  &u_h^+,u_h^-,q'_h\in\mathbb{Z}_{\geq0}
  &&\forall h\in\Hh.
  \end{aligned}
  \label{eq:variable-model}
\end{equation}

Suppose that the modified capacities satisfy \(0\leq q'_h\leq U_h\)
for every \(h\in\Hh\), where \(U_h\) is a given upper bound.
Each nonlinear stability inequality can be written as \(q'_h(1-z)\leq R\),
where \(z\) is a nonnegative integer expression
in the assignment and auxiliary variables and \(R\geq0\). It is equivalent
to \(q'_h-U_hz\leq R\). If \(z=0\), both inequalities require
\(q'_h\leq R\). If \(z\geq1\), both hold automatically because
\(0\leq q'_h\leq U_h\) and \(R\geq0\). 
It extends the bounded capacity linearization of
\citet{bobbio2026} to the three BIS blocking types.

For S1, the resulting inequality is
\begin{equation}
  q'_h
  -U_h\sum_{g\in H_h^{\succeq}(i)}x_{ig}
  \leq
  \sum_{j\in\D_i^{\succeq}(h)}x_{jh}
  \qquad \forall i\in\Sg,\ h\in H(i).
  \label{eq:var-S1}
\end{equation}
For S2, applying the same transformation to the two hospital
constraints gives
\begin{align}
  q'_h
  -U_h\left(
    \sum_{p'\in\Pk_p^{\succeq}(k)}y_{kp'}
    +\alpha_{kp}^{1}
  \right)
  &\leq
  \sum_{j\in\D_{d_k^1}^{\succeq}(h)}x_{jh}-x_{d_k^1h}
  &&
  \forall k\in\C,\ p=(h,g)\in\Pk_k,\ h\neq g,
  \label{eq:var-S2a}\\
  q'_g
  -U_g\left(
    \sum_{p'\in\Pk_p^{\succeq}(k)}y_{kp'}
    +\alpha_{kp}^{2}
  \right)
  &\leq
  \sum_{j\in\D_{d_k^2}^{\succeq}(g)}x_{jg}-x_{d_k^2g}
  &&
  \forall k\in\C,\ p=(h,g)\in\Pk_k,\ h\neq g.
  \label{eq:var-S2b}
\end{align}

For S3, consider \(p=(h,h)\in\Pk_k\), and relabel the members as
\(i_1,i_2\) so that \(i_1\succ_h i_2\). The linear constraint is
\begin{equation}
  q'_h
  -U_h\sum_{p'\in\Pk_p^{\succeq}(k)}y_{kp'}
  -1+x_{i_1h}+x_{i_2h}
  +\sum_{\ell\in\Kcross_h(k,i_2)}y_{\ell,(h,h)}
  \leq
  \sum_{j\in\D_{i_2}^{\succeq}(h)}x_{jh},
  \forall\, k\in\C,\quad
  p=(h,h)\in\Pk_k,\quad
  i_1\succ_h i_2.
  \label{eq:var-S3}
\end{equation}
The condition \(R\geq0\) is immediate for S1 and S2 because their
right-hand sides count assignees. For S3, each positive term on the left
of the following inequality identifies a distinct assignee whom \(h\)
ranks at least as highly as \(i_2\):
\[
  x_{i_1h}+x_{i_2h}
  +\sum_{\ell\in\Kcross_h(k,i_2)}y_{\ell,(h,h)}
  \leq
  \sum_{j\in\D_{i_2}^{\succeq}(h)}x_{jh}.
\]
Moving these correction terms to the right in the S3 constraint therefore
gives \(R\geq1\), as required for the transformation.

Under these bounds, the linear formulation is obtained from
\eqref{eq:variable-model} by replacing the constraints corresponding to
\eqref{eq:fixed-S1}, \eqref{eq:fixed-S2a}, \eqref{eq:fixed-S2b},
and \eqref{eq:fixed-S3} with
\eqref{eq:var-S1}, \eqref{eq:var-S2a}, \eqref{eq:var-S2b},
and \eqref{eq:var-S3}, respectively.

\subsection{Capacity modification policies}
\label{subsec:special-capacity-modifications}

The three policies below retain the assignment constraints, capacity
feasibility, linear stability constraints, and variable domains from
Section~\ref{subsec:variable-ip}. They differ only in their capacity
restrictions and objectives. The capacity upper bounds are \(U_h=q_h+1\)
for \ICMBC{} and \ICMUC{}, and \(U_h=q_h+2\) for \ICMNF{}.

Balanced capacity modification (\ICMBC{}) preserves total capacity and
allows each hospital to gain or lose at most one position. Add the restrictions
\begin{equation}
  |q'_h-q_h|\leq1\quad\forall h\in\Hh,
  \qquad
  \sum_{h\in\Hh}q'_h=\sum_{h\in\Hh}q_h.
  \label{eq:balanced-capacity-region}
\end{equation}
In the change variables, impose \(u_h^++u_h^-\leq1\) for every hospital
and \(\sum_hu_h^+=\sum_hu_h^-\). The objective is
\(\min\sum_hu_h^+\), which minimizes the number of transfers
\(T(q')=\frac12\sum_{h\in\Hh}|q'_h-q_h|\).
Each transfer removes one position from a hospital and adds it to another.

Unit capacity expansion (\ICMUC{}) instead prohibits reductions and allows
each hospital to gain at most one position. Given an expansion budget
\(B\in\{0,\ldots,|\Hh|\}\), the capacity restrictions are
\begin{equation}
  0\leq q'_h-q_h\leq1\quad\forall h\in\Hh,
  \qquad
  \sum_{h\in\Hh}(q'_h-q_h)\leq B.
  \label{eq:unit-expansion-capacity-region}
\end{equation}
These restrictions are imposed by setting \(u_h^-=0\),
\(u_h^+\in\{0,1\}\), and \(\sum_hu_h^+\leq B\).
For a given \(B\), retain the cardinality objective of
\eqref{eq:variable-model} to find a largest stable matching within the budget.
To find the fewest additions needed for stability, omit the budget
constraint and minimize \(\sum_hu_h^+\) instead.

Near feasible capacity modification (\ICMNF{}) permits each hospital to
gain or lose at most two positions while total capacity is preserved or
increases by one. Add the restrictions
\begin{equation}
  |q'_h-q_h|\leq2\quad\forall h\in\Hh,
  \qquad
  0\leq\sum_{h\in\Hh}(q'_h-q_h)\leq1.
  \label{eq:near-balanced-capacity-region}
\end{equation}
The objective is \(\min\sum_h(u_h^+-u_h^-)\).
Its optimal value is zero if stability can be achieved while preserving
total capacity, and one otherwise.
The following theorem establishes the exactness of these formulations.
\begin{theorem}
\label{thm:formulation-exactness}
ICM \eqref{eq:variable-model} is feasible if and only if the instance admits
a stable matching under some capacity vector \(q'\). The same equivalence
holds for \ICMBC{}, \ICMUC{}, and \ICMNF{} when \(q'\) is restricted to
\eqref{eq:balanced-capacity-region},
\eqref{eq:unit-expansion-capacity-region}, and
\eqref{eq:near-balanced-capacity-region}, respectively.
If the cardinality objective is used, it maximizes the number of matched
doctors over the corresponding stable matchings and capacity vectors.
\end{theorem}

The proof is given in Appendix~\ref{app:formulation-exactness-proof}.

\section{Cutoff formulations}
\label{sec:model-improvements}
\label{subsec:cutoff-formulations}

This section gives cutoff formulations of IFC, \ICMNF{}, \ICMBC{}, and
\ICMUC{}. 
For each hospital \(h\), assign consecutive positive integer scores
\(s_{ih}\) to its acceptable doctors, with a larger score indicating higher
priority. Thus, \(i\succ_h j\) if and only if \(s_{ih}>s_{jh}\). Scores are
compared only within the same hospital. Since \(s_{ih}\leq|\D|\), let
\(M_s:=|\D|+1\). For every hospital, introduce a continuous cutoff \(t_h\)
and a binary variable \(z_h\). In IFC, \(q'_h=q_h\) and \(U_h=q_h\).
In \ICMNF{}, \(q'_h\) is the modified capacity and \(U_h=q_h+2\);
in \ICMBC{} and \ICMUC{}, \(U_h=q_h+1\).

First, every doctor assigned to a hospital must meet its cutoff:
\begin{equation}
  t_h\leq s_{ih}+(M_s-s_{ih})(1-x_{ih})
  \qquad\forall(i,h)\in\E.
  \label{eq:cutoff-assignee}
\end{equation}
When \(x_{ih}=1\), this constraint reduces to \(t_h\leq s_{ih}\). When
\(x_{ih}=0\), it imposes no restriction beyond \(t_h\leq M_s\).

It remains to ensure that a hospital with an unfilled position cannot use a
positive cutoff. We impose
\begin{equation}
  q'_h-\sum_{i:(i,h)\in\E}x_{ih}
  \leq U_h(1-z_h),\qquad
  0\leq t_h\leq M_sz_h,\quad z_h\in\{0,1\},
  \qquad \forall h\in\Hh.
  \label{eq:cutoff-fullness}
\end{equation}
If \(h\) has an unfilled position, the first inequality forces \(z_h=0\),
and the second then gives \(t_h=0\). Hence such a hospital cannot reject an
acceptable doctor through a positive cutoff. When \(q'_h=0\), the hospital
has no position to fill, so \(z_h=1\) and \(t_h=M_s\) remain available.

For each single doctor \(i\), write the hospitals in \(H(i)\) as
$h_i^1\succ_i h_i^2\succ_i\cdots\succ_i h_i^{\ell_i}$,
and let \(\rho_i(h)\) be the position of \(h\) in this list. Introduce
continuous prefix variables \(b_{ir}\) and impose
\begin{equation*}
  b_{i0}=0,
  \qquad
  b_{ir}=b_{i,r-1}+x_{i h_i^r}
  \quad r=1,\ldots,\ell_i.
\end{equation*}
At every integer solution,
\(b_{i,\rho_i(h)}=1\) if and only if \(i\) receives \(h\) or a hospital
that \(i\) prefers to \(h\).

For a single doctor, S1 is imposed by
\begin{equation}
  s_{ih}+1
  \leq
  t_h+M_s b_{i,\rho_i(h)}
  \qquad\forall i\in\Sg,\ h\in H(i).
  \label{eq:cutoff-S1}
\end{equation}
If doctor \(i\) has not obtained \(h\) or a better hospital,
\(b_{i,\rho_i(h)}=0\) and the cutoff of \(h\) must be higher than \(i\)'s
score. Otherwise, the constraint is automatically satisfied.

Write the options of couple \(k\) as
$p_k^1\succ_k p_k^2\succ_k\cdots\succ_k p_k^{m_k}$.
Introduce continuous prefix variables \(a_{kt}\) and impose
\begin{equation*}
  a_{k0}=0,
  \qquad
  a_{kt}=a_{k,t-1}+y_{k p_k^t}
  \quad t=1,\ldots,m_k.
\end{equation*}
At every integer solution, \(a_{kt}=1\) if and only if couple \(k\)
receives \(p_k^t\) or a better option.

Now consider a couple \(k=(d_k^1,d_k^2)\) and
\(p=p_k^t=(h,g)\) with \(h\neq g\). The two hospital conditions for S2 are
\begin{align}
  s_{d_k^1h}+1
  &\leq
  t_h+M_s\left(
    a_{kt}+\alpha_{kp}^{1}
  \right),
  &&\forall k\in\C,\ t=1,\ldots,m_k,\ p=p_k^t=(h,g),\ h\neq g,
  \label{eq:cutoff-S2a}\\
  s_{d_k^2g}+1
  &\leq
  t_g+M_s\left(
    a_{kt}+\alpha_{kp}^{2}
  \right),
  &&\forall k\in\C,\ t=1,\ldots,m_k,\ p=p_k^t=(h,g),\ h\neq g.
  \label{eq:cutoff-S2b}
\end{align}
These inequalities are used together with \eqref{eq:fixed-S2c}. If the couple
prefers \((h,g)\) to its current assignment, \(a_{kt}=0\).
Constraint~\eqref{eq:fixed-S2c} ensures that at least one of the two
\(\alpha\) variables must then be zero, so the corresponding hospital must
have a cutoff higher than the score of the proposed member. If that member is
already assigned to the hospital, \eqref{eq:cutoff-assignee} prevents the
hospital from rejecting the member.

For S3, list the acceptable doctors of each hospital \(h\) according to its
priority order,
$i_h^1\succ_h i_h^2\succ_h\cdots\succ_h i_h^{m_h}$,
and let \(r_h(i)\) be the position of doctor \(i\) in this order. Introduce
continuous tail variables \(w_{hr}\) satisfying
\begin{equation*}
	w_{h0}=q'_h,
	\qquad
	w_{hr}=w_{h,r-1}-x_{i_h^r h}
	\quad r=1,\ldots,m_h.
\end{equation*}
It follows that
\begin{equation*}
	w_{h,r_h(i)}
	=q'_h-\sum_{j\succeq_h i}x_{jh}.
\end{equation*}
This is the number of unfilled positions plus the number of assignees whom
\(h\) ranks below \(i\). If unfilled positions are viewed as zero-score
positions, \(w_{h,r_h(i)}\geq r\) means that \(i\) passes the \(r\)th lowest
marginal cutoff.

Consider \(p_k^t=(h,h)\), and relabel its members so that
\(i_1\succ_h i_2\). The following condition excludes a blocking coalition
through any of conditions~\textnormal{(S3)(i)}, \textnormal{(S3)(ii)},
\textnormal{(S3)(iii)(a)}, and \textnormal{(S3)(iii)(c)} in
Definition~\ref{def:bis-blocking}:
\begin{equation}
	w_{h,r_h(i_2)-1}+x_{i_1h}
	\leq 1+(U_h-1)a_{kt}.
	\label{eq:cutoff-S3-base}
\end{equation}
The remaining condition, \textnormal{(S3)(iii)(b)}, depends on whether a
doctor ranked below \(i_2\) is assigned to \(h\) together with that doctor's
partner. To exclude this condition, introduce a continuous variable \(t_h^L\)
for each hospital, with \(0\leq t_h^L\leq M_s\). For every couple \(\ell\)
having the option
\((h,h)\), let
$\underline s_{\ell h}=\min\{s_{d_\ell^1h},s_{d_\ell^2h}\}$,
and impose
\begin{align}
	&t_h^L
	\leq \underline s_{\ell h}
	+(M_s-\underline s_{\ell h})(1-y_{\ell,(h,h)}),
	\label{eq:linked-cutoff-assignee}\\
	&s_{i_2h}+1
	\leq t_h^L+(s_{i_2h}+1)a_{kt}
	\qquad\forall k\in\C,\ t=1,\ldots,m_k,\ p_k^t=(h,h).
	\label{eq:linked-cutoff-candidate}
\end{align}

Together with constraints \eqref{eq:doctor}\text{--}\eqref{eq:link2},
\eqref{eq:fixed-S2c}, \eqref{eq:variable-capacity}, and
\eqref{eq:cutoff-assignee}--\eqref{eq:linked-cutoff-candidate}, these
constraints define the cutoff formulation. Setting \(q'=q\) gives
\IFCCUT{}. Combining the same stability constraints with the near feasible
capacity modification, balanced capacity modification, or unit capacity
expansion constraints gives \ICMCUTNF{}, \ICMCUTBC{}, or \ICMCUTUC{},
respectively.

\begin{theorem}
\label{thm:cutoff-equivalence}
For each of IFC, \ICMNF{}, \ICMBC{}, and \ICMUC{}, the original and cutoff
formulations admit exactly the same matching and capacity vector.
\end{theorem}

The proof is given in Appendix~\ref{app:cutoff-equivalence-proof}.
Appendix~\ref{app:alternative-S3-cutoffs} gives the formulations for the other three stability notions.

\section{Computational study}
\label{sec:computation}

We use Exact as the computational label for IFC and for ICM models using the
exact linearization, and Cutoff for their cutoff counterparts. 
The timing and scalability comparisons cover fixed capacities,
balanced capacity modification, and unit capacity expansion.
Near feasible capacity modification is examined separately in
Section~\ref{subsec:n100-experiments}.

We first validate IFC against the archive of \citet{delorme2021}
and compare Exact and Cutoff at fixed capacities, with Enumeration
as an additional benchmark for minimum capacity modification.
We also examine whether the formulation comparison changes under
the other stability notions.
Using Cutoff, we then study existence and the capacity modifications
required to restore stability across market types, followed by the
effect of expansion on maximum matching cardinality in initially
stable markets. Finally, we compare the solution rates of Exact
and Cutoff as market size increases.

All computations reported in this section were performed on a MacBook Pro with a 10-core Apple M2 Pro processor and 16 GB of memory, running macOS 26.5.2. The models were implemented in MATLAB R2026a and solved with Gurobi Optimizer 13.0.2 through its MATLAB interface.

\subsection{Formulation comparison}
\label{subsec:delorme-capacity-comparison}

We use the instances and computational records released by
\citet{delorme2021}. The \(F_{\mathrm{BIS}}\) output archive covers twenty
instance families with thirty instances each. Eight instance families permit a member of
a couple to be assigned while the other remains unmatched and therefore fall
outside our model. The remaining twelve instance families contain 360 source instances with
\(n\in\{20,750,1500,7500\}\).

We first validate our MATLAB implementation of IFC by comparing its results with the archived results generated by the original C++ implementation of \(F_{\mathrm{BIS}}\).
For all 360 instances, the numbers of variables, constraints, and
nonzero coefficients match the archived records exactly.
We reproduce all 334 reported optimal objective values.
For the twelve archived runs with unresolved optimality gaps,
our solution values lie within the reported bounds.
We also prove infeasibility for exactly the fourteen instances
recorded as \((-1,-1)\) in the archive.

For the timing comparison, hospital priorities that are already strict are
retained unchanged. When the source data contain hospital ties, we obtain
strict priorities in three ways. The first orders doctors within each tie by
identifier. The second uses one random lottery consistently at every
hospital, and the third uses an independent lottery at each hospital.
Strict comparisons between tie classes are never changed. 

The source instances construct each couple's joint list as the Cartesian
product of the two individual lists.
Options with individual
ranks \(r\) and \(s\) are ordered first by \(r+s\) and then by
\(\max\{r,s\}\), both in increasing order. Options equal on both
criteria are tied. Families restricted to assignments at the same
hospital retain only pairs \((h,h)\).
For the timing comparison, we break remaining couple ties in source
order, preserving all strict comparisons in the original preferences.
Duplicate source markets are counted only once in the timing
sample.\footnote{The families \texttt{ML\_10\_AS} and
	\texttt{ML\_10} contain identical instance files, so their thirty
	instances are counted once.}
This leaves 240 instances with strict hospital priorities and 90
with hospital ties. Applying the three hospital tie treatments
to the latter gives \(240+3\times90=510\) strict variants.

Before building either fixed capacity formulation for the Delorme
instances, we apply the \(F_{\mathrm{BIS}}\) preference list reduction
of \citet{delorme2021}. The hospital offer and resident application
rules alternate until no further deletions occur, removing rejected
doctor--hospital pairs and the couple options containing them.
Enumeration repeats this reduction for each candidate capacity vector.
This preprocessing is applied only when hospital capacities are fixed.

\begin{table}
	\centering
	\begin{threeparttable}
		\caption{Fixed capacity comparison of IFC and \IFCCUT{}}
		\label{tab:delorme-fixed-comparison}
		\scriptsize
		\setlength{\tabcolsep}{1.4pt}
		\begin{tabular*}{\textwidth}{@{\extracolsep{\fill}}rlllrrrrrrr@{}}
			\toprule
			& & & & & \multicolumn{2}{c}{Mean total time (s)}
			& & \multicolumn{3}{c}{\IFCCUT{} relative to IFC} \\
			\cmidrule(lr){6-7}\cmidrule(lr){9-11}
			\(n\) & Hospital priorities & Couple options & Tie breaking & Stable/inf.
			& IFC & \IFCCUT{} & Ratio & Variables & Constraints & Nonzeros \\
			\midrule
			20 & Master list & One hospital & None & 30/0
			& 0.006 & 0.009 & 0.77 & \(+4.2\%\) & \(+40.3\%\) & \(+41.5\%\) \\
			20 & Master list & One hospital & ID & 30/0
			& 0.011 & 0.017 & 0.73 & \(+7.2\%\) & \(+38.6\%\) & \(+42.7\%\) \\
			20 & Master list & One hospital & Common lottery & 30/0
			& 0.013 & 0.016 & 0.80 & \(+5.3\%\) & \(+39.5\%\) & \(+41.6\%\) \\
			20 & Master list & One hospital & Independent & 30/0
			& 0.014 & 0.020 & 0.74 & \(+4.3\%\) & \(+39.7\%\) & \(+40.6\%\) \\
			\addlinespace
			750 & Hospital specific & General & None & 24/6
			& 2.118 & 2.016 & 1.05 & \(-22.9\%\) & \(+9.9\%\) & \(-2.2\%\) \\
			750 & Hospital specific & General & ID & 28/2
			& 2.181 & 1.983 & 1.10 & \(-22.9\%\) & \(+9.9\%\) & \(-2.2\%\) \\
			750 & Hospital specific & General & Common lottery & 30/0
			& 1.934 & 1.952 & 0.98 & \(-22.8\%\) & \(+10.0\%\) & \(-2.0\%\) \\
			750 & Hospital specific & General & Independent & 30/0
			& 2.044 & 1.937 & 1.03 & \(-22.8\%\) & \(+10.0\%\) & \(-2.1\%\) \\
			750 & Master list & General & None & 30/0
			& 1.417 & 1.727 & 0.84 & \(-22.7\%\) & \(+10.1\%\) & \(-1.9\%\) \\
			750 & Master list & General & ID & 26/4
			& 1.240 & 1.286 & 0.97 & \(-22.5\%\) & \(+10.2\%\) & \(-1.7\%\) \\
			750 & Master list & General & Common lottery & 27/3
			& 1.444 & 1.504 & 0.97 & \(-22.7\%\) & \(+10.1\%\) & \(-1.9\%\) \\
			750 & Master list & General & Independent & 27/3
			& 1.388 & 1.493 & 0.94 & \(-22.7\%\) & \(+10.1\%\) & \(-1.9\%\) \\
			750 & Master list & One hospital & None & 30/0
			& 0.264 & 0.306 & 0.86 & \(-7.3\%\) & \(+40.9\%\) & \(+26.3\%\) \\
			\addlinespace
			1,500 & Hospital specific & General & None & 30/0
			& 3.830 & 5.098 & 0.75 & \(-22.6\%\) & \(+9.5\%\) & \(-2.3\%\) \\
			1,500 & Master list & General & None & 28/2
			& 3.379 & 3.682 & 0.93 & \(-22.4\%\) & \(+9.8\%\) & \(-2.0\%\) \\
			7,500 & Hospital specific & General & None & 29/1
			& 30.081 & 35.997 & 0.84 & \(-22.4\%\) & \(+9.3\%\) & \(-2.4\%\) \\
			7,500 & Master list & General & None & 25/5
			& 25.728 & 32.672 & 0.79 & \(-22.3\%\) & \(+9.5\%\) & \(-2.1\%\) \\
			\midrule
			\multicolumn{4}{l}{All 510 variants} & 484/26
			& 4.535 & 5.395 & 0.88 & \(-22.3\%\) & \(+9.9\%\) & \(-1.9\%\) \\
			\bottomrule
		\end{tabular*}
		\begin{tablenotes}[flushleft]
			\footnotesize
			\item Stable/inf.\ gives the common existence classification. Ratio is the
			paired geometric mean of the IFC to \IFCCUT{} running-time ratio.
		\end{tablenotes}
	\end{threeparttable}
\end{table}

We compare IFC and \IFCCUT{} at fixed capacities, with both formulations
maximizing the number of matched doctors. Reported times include
preprocessing, model construction, and optimization.
The formulations agree on the optimal value for each of the 484
feasible variants and prove the remaining 26 infeasible.
Table~\ref{tab:delorme-fixed-comparison} groups the results by
hospital priorities in the source data, couple options, and the
tie-breaking rule, with thirty variants per row.
Master list means that
hospitals share a ranking of doctors, whereas Hospital specific means
that each hospital has its own ranking.
One hospital means that each couple accepts only assignments placing
both partners at the same hospital.
Tie breaking specifies how ties in hospital priorities are resolved.

Compared with IFC, \IFCCUT{} uses, on average, 22.3\% fewer variables
and 1.9\% fewer nonzero coefficients, but 9.9\% more constraints.
IFC is modestly faster overall, with a paired geometric mean ratio
of IFC time to \IFCCUT{} time of 0.88.
At the group level, \IFCCUT{} is faster in three of the four groups
with \(n=750\), hospital specific priorities, and general couple
options, with ratios of 1.03--1.10. IFC is faster in all remaining
groups.
Hospital tie breaking can affect existence. In the tied \(n=750\)
family with hospital specific priorities, identifier ordering yields
28 variants with stable matchings, compared with 30 under either
lottery. In the corresponding master list family with general couple
options, the counts are 26 and 27, respectively.

We next consider the 26 variants proved infeasible at the original
capacities: 18 with \(n=750\), two with \(n=1500\), and six with
\(n=7500\). The objective is to minimize the number of transfers
under balanced modification and the number of added positions
under unit expansion.

We compare Enumeration, Exact, and Cutoff.
Enumeration searches capacity vectors in increasing order of the
number of added positions or transfers.
For unit expansion, it tests \(q'=q+\mathbf{1}_S\) for every
\(S\subseteq\Hh\) with \(|S|=B\), increasing \(B\) from one.
For balanced modification, it tests
\(q'=q-\mathbf{1}_D+\mathbf{1}_R\) for disjoint sets
\(D,R\subseteq\Hh\) with \(|D|=|R|=T\) and \(q_h>0\) for all
\(h\in D\), increasing \(T\) from one.
At each candidate vector, IFC is solved with a zero objective
to test feasibility.
Enumeration proceeds to a larger repair level only after all
candidates at the current level are proved infeasible.
The first independently verified stable matching therefore certifies
the minimum repair.

Exact uses the linearized \ICMBC{} and \ICMUC{} models, whereas
Cutoff uses \ICMCUTBC{} and \ICMCUTUC{}. Both use the same valid
inequalities. Each method has a total time limit of 3,600 seconds
per instance and policy. All runs use one Gurobi thread, Method~2,
zero MIP gap, and the default Seed setting. Total time includes
preprocessing, model construction, optimization, and verification.

We use penalized average runtime \citep{hutter2009paramils},
with a penalty factor of two (PAR-2).
For each method and policy, this averages over all 26 runs,
using the observed total time when the minimum repair is certified
and \(7{,}200\) seconds otherwise.
Table~\ref{tab:delorme-minimum-modification} summarizes the 156 runs.
One balanced transfer suffices in every variant. Under unit expansion,
25 variants require one added position, and one \(n=7500\) master
list variant requires two.

\begin{table}[H]
	\centering
	\begin{threeparttable}
		\caption{Minimum capacity modification on the 26 fixed capacity infeasible
			variants}
		\label{tab:delorme-minimum-modification}
		\scriptsize
		\setlength{\tabcolsep}{3.4pt}
		\begin{tabular*}{\textwidth}{@{\extracolsep{\fill}}llrrrrrrr@{}}
			\toprule
			& & & &
			\multicolumn{2}{c}{Minimum modification}
			& \multicolumn{3}{c}{Total time (s)} \\
			\cmidrule(lr){5-6}\cmidrule(lr){7-9}
			Policy & Method & Optimal & Unknown & 1 & 2
			& Mean & Median & PAR-2 \\
			\midrule
			\ICMBC{} & Enumeration & 24 & 2 & 24 & 0 & 193.40 & 6.07 & 732.37 \\
			& Exact       & 25 & 1 & 25 & 0 & 80.08  & 40.57 & 353.92 \\
			& Cutoff      & 26 & 0 & 26 & 0 & 65.17  & 10.48 & 65.17 \\
			\addlinespace
			\ICMUC{} & Enumeration & 24 & 2 & 24 & 0 & 194.64 & 9.75 & 733.51 \\
			& Exact       & 26 & 0 & 25 & 1 & 47.55  & 9.72 & 47.55 \\
			& Cutoff      & 26 & 0 & 25 & 1 & 35.86  & 5.85 & 35.86 \\
			\bottomrule
		\end{tabular*}
		\begin{tablenotes}[flushleft]
			\footnotesize
			\item Optimal denotes a verified stable matching together with proof that no
			smaller modification is feasible. 
		\end{tablenotes}
	\end{threeparttable}
\end{table}

Cutoff certifies the minimum repair in all 26 variants under both
policies. Exact certifies 25 under balanced modification and 26 under
unit expansion, while Enumeration certifies 24 under each policy.
The ratios of mean PAR-2 time for Exact relative to Cutoff are
5.43 under balanced modification and 1.33 under unit expansion.
For Enumeration relative to Cutoff, the corresponding ratios are
11.24 and 20.45.
Unlike the fixed capacity comparison, these results favor Cutoff
for minimum capacity modification on this benchmark, although it
is not faster on every instance.

\subsection{Existence and capacity modification across stability conditions}
\label{subsec:alternative-stability-experiment}

To examine how the stability notion affects existence and capacity
modification, we use three types of markets released by
\citet{delorme2021}. Each type contains 30 markets, for a total of 90.
Each market has 750 doctors, 50 hospitals,
and 75 couples. The first type uses hospital specific priorities
and general joint couple lists. The second uses a master list and
general joint lists. The third also uses a master list, but every
acceptable couple option places both partners at the same hospital.
Hospital priorities are strict in all three types. In the first
two types, ties in couple preferences are broken in source order.
The third type already has strict couple preferences.

Using the same markets, we compare existence and computational
performance under BIS, KPR, MM, and \NV{} stability.
We first compare IFC and \IFCCUT{} at fixed capacities, examining
how their relative performance varies with the stability notion
and market type. When nonexistence is proved under a given notion,
we compare Enumeration, Exact, and Cutoff in solving the
corresponding capacity modification problems. The objectives
minimize the number of balanced transfers under \ICMBC{} and
the number of added positions under \ICMUC{}. We evaluate the
methods by their ability to certify these minima and their
running times.

All runs use one Gurobi thread, Method~2, zero MIP gap, and default
Seed setting.
Each method has a time limit of 3,600 seconds per problem.
In the first stage, IFC and \IFCCUT{} maximize matching cardinality
for each of the 90 markets under all four stability notions,
yielding 720 runs and 360 paired comparisons. All runs end with
a proof of optimality or infeasibility. The formulations agree
on existence in every comparison and on maximum cardinality
whenever a stable matching exists.
Table~\ref{tab:alternative-stability-fixed} summarizes existence
and running times.

\begin{table}[H]
	\centering
	\begin{threeparttable}
		\caption{Fixed capacity existence and running times across stability conditions}
		\label{tab:alternative-stability-fixed}
		\scriptsize
		\setlength{\tabcolsep}{2.0pt}
		\begin{tabular*}{\textwidth}{@{\extracolsep{\fill}}lllrrrrrr@{}}
			\toprule
			& & & & \multicolumn{2}{c}{Existence} &
			\multicolumn{2}{c}{Mean total time (s)} & \\
			\cmidrule(lr){5-6}\cmidrule(lr){7-8}
			Stability & Hospital priorities & Couple options & Markets
			& Stable & Infeasible & IFC & \IFCCUT{} & Ratio \\
			\midrule
			BIS & Hospital specific & General      & 30 & 24 & 6  & 2.118 & 2.016 & 1.05 \\
			& Master list       & General      & 30 & 30 & 0  & 1.417 & 1.727 & 0.84 \\
			& Master list       & One hospital & 30 & 30 & 0  & 0.264 & 0.306 & 0.86 \\
			\cmidrule(lr){2-9}
			& All               &              & 90 & 84 & 6  & 1.266 & 1.349 & 0.91 \\
			\addlinespace
			KPR & Hospital specific & General      & 30 & 24 & 6  & 2.390 & 2.285 & 1.02 \\
			& Master list       & General      & 30 & 30 & 0  & 1.537 & 1.747 & 0.90 \\
			& Master list       & One hospital & 30 & 30 & 0  & 0.277 & 0.329 & 0.84 \\
			\cmidrule(lr){2-9}
			& All               &              & 90 & 84 & 6  & 1.401 & 1.454 & 0.91 \\
			\addlinespace
			MM  & Hospital specific & General      & 30 & 24 & 6  & 4.111 & 2.366 & 1.08 \\
			& Master list       & General      & 30 & 30 & 0  & 1.583 & 1.862 & 0.86 \\
			& Master list       & One hospital & 30 & 0  & 30 & 0.257 & 0.312 & 0.83 \\
			\cmidrule(lr){2-9}
			& All               &              & 90 & 54 & 36 & 1.984 & 1.513 & 0.92 \\
			\addlinespace
			\NV{}                 & Hospital specific & General      & 30 & 24 & 6  & 2.428 & 2.299 & 1.03 \\
			& Master list       & General      & 30 & 29 & 1  & 1.505 & 1.718 & 0.89 \\
			& Master list       & One hospital & 30 & 0  & 30 & 0.223 & 0.309 & 0.72 \\
			\cmidrule(lr){2-9}
			& All               &              & 90 & 53 & 37 & 1.385 & 1.442 & 0.87 \\
			\bottomrule
		\end{tabular*}
	\end{threeparttable}
\end{table}

The main difference in existence concerns the 30 master list markets
where couples accept only assignments to the same hospital.
All 30 admit a stable matching under BIS and KPR, but none does
under MM or \NV{}. One master list market with general joint lists
is also infeasible under \NV{}, but feasible under the other three
notions. The same six markets with hospital specific priorities
are infeasible under all four notions. Overall, BIS and KPR each
admit a stable matching in 84 markets, compared with 54 under MM
and 53 under \NV{}.

Within each market type, the relative ranking of IFC and \IFCCUT{}
is unchanged across the four stability notions.
With hospital specific priorities, \IFCCUT{} is modestly faster,
with paired geometric mean ratios of 1.02--1.08.
Under a master list, IFC is faster, with ratios of 0.84--0.90
for general joint lists and 0.72--0.86 when couples accept only
assignments to the same hospital.
Across all 360 comparisons, the paired geometric mean ratio is
0.90, indicating a modest advantage for IFC under this measure.

In the second stage, we consider the 85 market and stability
combinations for which both fixed capacity formulations prove
infeasibility: six under BIS, six under KPR, 36 under MM, and
37 under \NV{}.
For each combination, we minimize the number of transfers under
balanced modification and the number of added positions under
unit expansion, using Enumeration, Exact, and Cutoff
with a limit of 3,600 seconds per method and repair problem.
Applying three methods to 85 combinations under two capacity
policies gives \(85\times2\times3=510\) computational runs.

Exact and Cutoff certify the minimum repair in all 170 problems
and obtain identical minimum values. All returned matchings pass
the independent stability check.
Enumeration certifies 56 minima: all 24 under BIS and KPR,
15 of 72 under MM, and 17 of 74 under \NV{}.
The remaining 114 runs reach the time limit without finding
a stable matching.

For every selected BIS and KPR case, the minimum is one transfer
under balanced modification and one added position under unit
expansion. Under MM, the corresponding means are 2.56 transfers
and 5.17 added positions. Under \NV{}, they are 2.59 and 5.22.
For both MM and \NV{}, the largest minima are four transfers
and ten added positions. Certifying these larger minima requires
Enumeration to exclude all capacity vectors at smaller repair levels.

\begin{figure}
	\centering
	\includegraphics[width=0.90\textwidth]{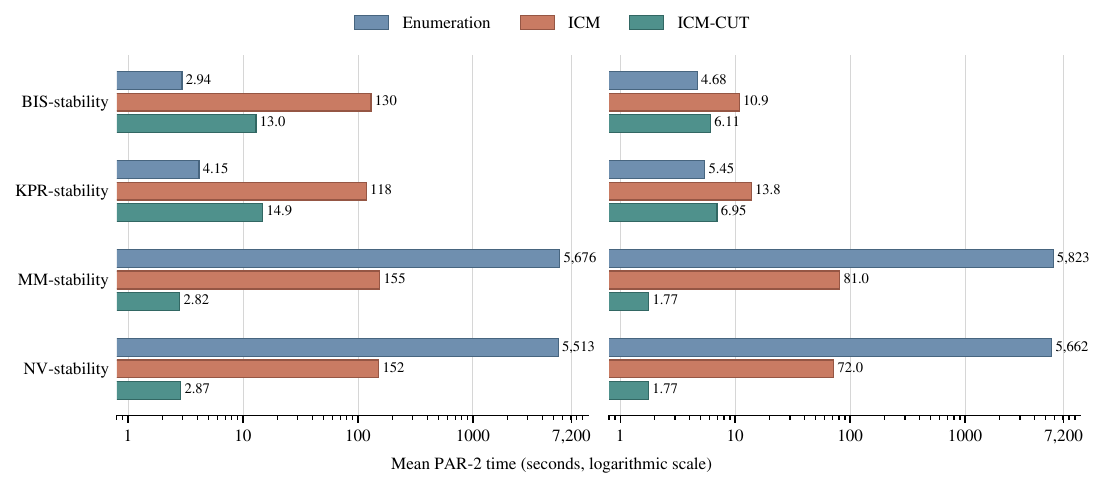}
	\vspace{-1.2em}
	\caption{Mean PAR-2 running times for minimum capacity modification under
		\ICMBC{} (left) and \ICMUC{} (right).}
	\label{fig:alternative-stability-capacity-runtime}
\end{figure}

Figure~\ref{fig:alternative-stability-capacity-runtime}
compares the three methods using mean PAR-2 times.
Enumeration has the lowest mean PAR-2 time under BIS and KPR,
finding a stable matching after testing at most three capacity
vectors under balanced modification and four under unit expansion.
Cutoff has the lowest mean PAR-2 time under MM and \NV{}.
For MM and \NV{}, respectively, the Exact to Cutoff ratios are
approximately 55 and 53 under balanced modification, and 46 and
41 under unit expansion.
Given its more consistent computational performance, we adopt
Cutoff as the primary formulation for subsequent experiments,
all of which use BIS stability.

\subsection{Existence and capacity modification across market types}
\label{subsec:n100-experiments}

This experiment examines how market structure affects the existence
of stable matchings and the minimum capacity modification needed
to restore stability. We vary hospital priorities, couple preferences,
whether partners may be assigned to the same hospital, the couple
proportion, and total capacity.

We generate 1,000 base profiles, each with 100 doctors and ten
hospitals. Each doctor independently draws a uniformly random
ordered list of six distinct hospitals and finds the remaining
four unacceptable. The couple proportion, defined as the share
of doctors who belong to a couple, takes target values of
5\%, 15\%, 25\%, 50\%, 75\%, and 100\%.
Within each profile, couples are formed from a common sequence
of disjoint doctor pairs. Increasing the proportion therefore
adds couples without changing existing pairs.

For each couple proportion and replication, we construct eight
market types with total capacity \(Q=100\). These types combine
two hospital priority structures (hospital specific priorities
or a master list), two couple preference rules (rank sum or
random order), and two treatments of assignments to the same
hospital (allowed or excluded).
The eight types share a capacity vector constructed by giving
each hospital one position and allocating the remaining 90
positions independently and uniformly across hospitals.

Each couple's initial joint options are obtained by filtering
the Cartesian product of the partners' individual lists through
a symmetric compatibility matrix. Every hospital is compatible
with itself, and each unordered pair of distinct hospitals is
independently declared compatible with probability 0.75.

The two preference rules order this same set of joint options.
Under rank sum preferences, an option \((h,g)\) for couple
\((a,b)\) is ranked by increasing
\(\operatorname{rank}_a(h)+\operatorname{rank}_b(g)\), where ranks
refer to the partners' individual lists.
Equal sums are broken by the smaller value of
\(\max\{\operatorname{rank}_a(h),\operatorname{rank}_b(g)\}\).
Any remaining tie favors the option giving the better individual
rank to the partner with the higher independently generated
auxiliary score.
Under random preferences, each couple instead receives a uniformly
random strict ordering of the same options.
When assignments to the same hospital are excluded, all options
\((h,h)\) are removed without changing the relative order of
the remaining options.

The design is paired. Within each replication, individual
preference lists are shared across all six couple proportions.
For each proportion, the eight types differ only in the three
experimental factors, with couple membership, compatibility,
capacities, and all other inputs held fixed.
Types~1--8 in Table~\ref{tab:n100-experiments-summary} therefore
contain \(6\times2\times2\times2\times1{,}000=48{,}000\) instances.
Types~9--14 extend the master list and rank sum cases to total
capacities \(Q\in\{90,95,99\}\), retaining both treatments of
assignments to the same hospital.
Each instance preserves all noncapacity inputs from its
corresponding \(Q=100\) market.
Each hospital still receives one initial position, and the
remaining positions are allocated using the first \(Q-10\)
draws from the same sequence of 90 hospital draws.
Thus \(q_h(90)\leq q_h(95)\leq q_h(99)\leq q_h(100)\)
for every hospital.
The three capacity levels, six couple proportions, two treatments,
and 1,000 replications yield another 36,000 instances,
bringing the total to 84,000.

\begin{table}
	\centering
	\begin{threeparttable}
		\caption{Existence and minimum capacity modification across
			the fourteen \(n=100\) instance types}
		\label{tab:n100-experiments-summary}
		\scriptsize
		\setlength{\tabcolsep}{1.4pt}
		\begin{tabular*}{\textwidth}{@{\extracolsep{\fill}}crllcr*{9}{r}@{}}
			\toprule
			& & & & & & \multicolumn{6}{c}{\makebox[0pt][c]{\shortstack{\ICMCUTBC{}\\Minimum transfers}}}
			& \multicolumn{3}{c}{\makebox[0pt][c]{\shortstack{\ICMCUTUC{}\\Minimum added positions}}} \\
			\cmidrule(lr){7-12}\cmidrule(lr){13-15}
			Type & \(Q\) & Hospital priorities & Couple preferences & \((h,h)\)
			& \shortstack{\IFCCUT{}\\stable} & \(T=1\) & \(T=2\) & \(T=3\)
			& \(T=4\) & \(T=5\) & Inf. & \(B=1\) & \(B=2\) & \(B=3\) \\
			\midrule
			1  & 100 & Hospital specific & Random order & No  & 5,793 & 207   & 0   & 0  & 0 & 0 & 0  & 207   & 0   & 0 \\
			2  & 100 & Hospital specific & Random order & Yes & 5,817 & 183   & 0   & 0  & 0 & 0 & 0  & 179   & 4   & 0 \\
			3  & 100 & Hospital specific & Rank sum     & No  & 5,800 & 200   & 0   & 0  & 0 & 0 & 0  & 200   & 0   & 0 \\
			4  & 100 & Hospital specific & Rank sum     & Yes & 5,834 & 166   & 0   & 0  & 0 & 0 & 0  & 165   & 1   & 0 \\
			5  & 100 & Master list       & Random order & No  & 5,152 & 848   & 0   & 0  & 0 & 0 & 0  & 834   & 14  & 0 \\
			6  & 100 & Master list       & Random order & Yes & 5,011 & 989   & 0   & 0  & 0 & 0 & 0  & 976   & 13  & 0 \\
			7  & 100 & Master list       & Rank sum     & No  & 5,380 & 620   & 0   & 0  & 0 & 0 & 0  & 610   & 10  & 0 \\
			8  & 100 & Master list       & Rank sum     & Yes & 5,316 & 684   & 0   & 0  & 0 & 0 & 0  & 672   & 12  & 0 \\
			\addlinespace
			9  & 90  & Master list       & Rank sum     & No  & 4,902 & 1,051 & 36  & 5  & 3 & 0 & 3  & 1,013 & 85  & 0 \\
			10 & 90  & Master list       & Rank sum     & Yes & 4,865 & 1,070 & 47  & 10 & 5 & 0 & 3  & 1,000 & 135 & 0 \\
			11 & 95  & Master list       & Rank sum     & No  & 4,708 & 1,182 & 100 & 10 & 0 & 0 & 0  & 1,264 & 28  & 0 \\
			12 & 95  & Master list       & Rank sum     & Yes & 4,480 & 1,309 & 167 & 24 & 6 & 1 & 13 & 1,484 & 35  & 1 \\
			13 & 99  & Master list       & Rank sum     & No  & 5,198 & 789   & 12  & 1  & 0 & 0 & 0  & 794   & 8   & 0 \\
			14 & 99  & Master list       & Rank sum     & Yes & 5,104 & 883   & 12  & 1  & 0 & 0 & 0  & 878   & 18  & 0 \\
			\midrule
			\multicolumn{5}{l}{All types} & 73,360 & 10,181 & 374 & 51 & 14 & 1 & 19
			& 10,276 & 363 & 1 \\
			\bottomrule
		\end{tabular*}
		\begin{tablenotes}[flushleft]
			\footnotesize
			\item The modification columns contain only instances for which
			\IFCCUT{} proves nonexistence. Inf. denotes proved infeasibility within the
			balanced unit modification region.
		\end{tablenotes}
	\end{threeparttable}
\end{table}

We first solve \IFCCUT{} to maximize matching cardinality at the
original capacities. For instances proved infeasible, we solve
\ICMCUTBC{} and \ICMCUTUC{} to minimize the number of balanced
transfers and added positions, respectively.
Each model has a time limit of 600 seconds.

For the same 10,640 instances, we also allow each hospital to
gain or lose up to two positions.
Using the model in
Section~\ref{subsec:special-capacity-modifications},
we minimize \(\sum_h(u_h^+-u_h^-)\), requiring this sum to be
nonnegative but removing the aggregate upper bound of one.
This determines whether stability can be restored while preserving
total capacity. The Cutoff formulation, solver settings, and
stability checks remain unchanged.

All 115,920 runs terminate with a proof of optimality or
infeasibility. At the original capacities, \IFCCUT{} finds a
maximum cardinality stable matching in 73,360 of the 84,000
instances and proves nonexistence in the remaining 10,640.
Its mean running time is 2.169 seconds.
For the 10,640 instances requiring repair, mean running times
are 2.073 seconds for \ICMCUTBC{} and 1.069 seconds for
\ICMCUTUC{}. The longest run takes 281.1 seconds.

Across the fourteen market types, the pooled nonexistence rate
increases with the target couple proportion, from 2.18\% at
5\% to 21.61\% at 100\%.
At \(Q=100\), nonexistence is more frequent under a master list
than under hospital specific priorities, at 13.09\% and 3.15\%,
respectively. The differences are smaller for the other two
design factors: 9.28\% under random couple preferences versus
6.96\% under rank sum preferences, and 8.43\% when assignments
to the same hospital are allowed versus 7.81\% when they
are excluded.

The capacity comparison is nonmonotone.
For the master list markets with rank sum preferences, the
numbers of infeasible instances are 2,233, 2,812, 1,698, and
1,304 at \(Q=90,95,99,\) and 100, respectively.
Each capacity level contains 12,000 paired instances.
In particular, 2,166 instances admit a stable matching at
\(Q=90\) but not at \(Q=95\), while 1,587 exhibit the reverse.

Across the eight \(Q=100\) market types, one balanced transfer
repairs all 3,897 instances infeasible at the original capacities.
At the lower capacity levels, \ICMCUTBC{} repairs 6,724 instances.
One transfer suffices in 6,284, while the remaining 440 require
two to five transfers.
Balanced unit modification is proved infeasible in 19 instances:
six at \(Q=90\) and thirteen at \(Q=95\).
By contrast, \ICMCUTUC{} repairs all 10,640 instances while
adding at most one position per hospital.
The minimum number of added positions is one in 10,276 instances,
two in 363, and three in one.
Each of the 19 instances infeasible under balanced modification
requires only one added position.

When each hospital may gain or lose up to two positions,
the minimum increase in total capacity is zero in all 10,640
instances. Thus, every tested instance can be repaired while
preserving total capacity, including the 19 instances infeasible
under balanced unit modification.

\subsection{Existence and capacity modification across methods}
\label{subsec:existence-comparison}

We next examine how the number of couples affects stable matching
existence in larger markets and determine the minimum capacity
modification required when stability fails.
We use the archived \texttt{data500} collection of
\citet{biro2011}. Each market contains 500 doctors, 50 hospitals,
and 500 positions. Every doctor lists six acceptable hospitals,
and distinct doctor scores induce a common master list.
The collection contains 1,000 instances for each of
0, 12, 25, 50, 75, 100, 125, 150, 175, 200, 225, and 250 couples,
giving 12,000 markets.
The authors generated 1,000 base markets with randomly allocated
hospital capacities and progressively paired more doctors into
couples. We retain the archived individual preferences, capacity
vectors, and couple identities.

The archived files do not contain couples' joint preferences.
We construct each couple's joint options from the Cartesian
product of its members' individual lists, retaining a hospital
pair only when its entry in the shared compatibility matrix
is one.
We order the retained options using the rank sum rule in
Section~\ref{subsec:n100-experiments}.
The final tie breaking uses the archived doctor scores.
The main experiment uses a time limit of 600 seconds
per model.

\begin{table}
	\centering
	\begin{threeparttable}
		\caption{Fixed capacity outcomes and certified minimum capacity
			modifications in the reconstructed \texttt{data500} markets}
		\label{tab:algorithm-comparison}
		\scriptsize
		\setlength{\tabcolsep}{1.6pt}
		\begin{tabular*}{\textwidth}{@{\extracolsep{\fill}}l*{11}{r}@{}}
			\toprule
			& \multicolumn{11}{c}{Number of couples} \\
			\cmidrule(lr){2-12}
			Method and outcome
			& 12 & 25 & 50 & 75 & 100 & 125 & 150 & 175 & 200 & 225 & 250 \\
			\midrule
			Roth--Peranson\textsuperscript{a}
			& 952 & 897 & 701 & 547 & 395 & 277 & 170 & 83 & 41 & 9 & 3 \\
			Best heuristics in B-I-S\textsuperscript{b}
			& 976 & 958 & 911 & 870 & 815 & 759 & 716 & 628 & 414 & 154 & 15 \\
			Scarf (int.\ solution)\textsuperscript{c}
			& 895 & 813 & 649 & 532 & 426 & 356 & 316 & 261 & 202 & 174 & 158 \\
			\IFCCUT{}
			& 975 & 959 & 911 & 869 & 821 & 781 & 755 & 725 & 695 & 672 & 647 \\
			\midrule
			\ICMCUTBC{} [\(T=1\)]
			& 25 & 41 & 89 & 131 & 179 & 219 & 245 & 275 & 305 & 328 & 353 \\
			\ICMCUTUC{} [\(B=1\)]
			& 25 & 41 & 89 & 121 & 169 & 211 & 231 & 264 & 294 & 316 & 335 \\
			\ICMCUTUC{} [\(B=2\)]
			& 0 & 0 & 0 & 10 & 9 & 8 & 14 & 11 & 11 & 12 & 18 \\
			\ICMCUTUC{} [\(B=3\)]
			& 0 & 0 & 0 & 0 & 1 & 0 & 0 & 0 & 0 & 0 & 0 \\
			\bottomrule
		\end{tabular*}
		\begin{tablenotes}[flushleft]
			\footnotesize
			\item[a] Values from Table II of \citet{biro2011}; the Roth--Peranson variant
			used a time limit of 5 seconds.
			\item[b] Best single heuristic variant in each column of Table VIII of
			\citet{biro2011}; each variant used a time limit of 50 seconds.
			\item[c] Integral Scarf solutions in Table 1 of \citet{biro2016}.
			\item The four capacity modification rows contain only markets for which
			\IFCCUT{} proves nonexistence.
		\end{tablenotes}
	\end{threeparttable}
\end{table}

Table~\ref{tab:algorithm-comparison} reports our fixed capacity
existence results and certified minimum repairs, alongside the
published counts of markets in which a stable matching was
found.
The published counts provide context rather than a direct
comparison, as the computational procedures and time limits
also differ.

At the original capacities, \IFCCUT{} finds a maximum cardinality
stable matching in 9,810 markets and proves nonexistence in the
remaining 2,190.
All markets without couples admit a stable matching.
The nonexistence rate rises from 2.5\% with 12 couples to
35.3\% with 250 couples.\footnote{For 12 and 75 couples, the published heuristic
	counts are 976 and 870, compared with 975 and 869 in our
	reconstruction. These small differences may reflect the
	construction of couples' joint preferences, which are not
	recorded in the archived files. We use the rank sum rule
	described in \citet{biro2011}, whereas an accompanying MATLAB
	script orders hospital pairs primarily by the larger of the
	partners' individual ranks. These rules can produce different
	joint rankings and thus change whether a stable matching exists.}
Both capacity policies restore stability in all 2,190 markets.
The minimum balanced repair is one transfer in every case.
Under unit expansion, the minimum number of added positions
is one in 2,096 markets, two in 93, and three in one.\footnote{
	Three \ICMCUTBC{} runs and one \ICMCUTUC{} run, all with
	250 couples, did not prove optimality within 600 seconds.
	Increasing the limit to 3,600 seconds with all other settings
	unchanged certified \(T=1\) for the three BC cases and \(B=2\)
	for the UC case. All timing statistics use the original
	600-second runs.}

Mean \IFCCUT{} time rises from 0.14 seconds without couples
to 13.57 seconds with 250 couples.
Across the 2,190 repair markets, mean PAR-2 times are
17.87 seconds for \ICMCUTBC{} and 4.88 seconds for
\ICMCUTUC{}, giving a BC to UC ratio of 3.66.
The UC model has a lower PAR-2 time in 2,172 of the 2,190
paired comparisons.

\subsection{Matching cardinality with additional positions}
\label{subsec:unit-expansion-cardinality}

The preceding experiment studies the minimum capacity modification
needed to restore stability. We now examine how a bounded expansion
budget can increase the maximum number of matched doctors in markets
that already admit a stable matching.
For market \(r\) with initial capacity vector \(q^r\), define
\[
m_r(B)=\max\Bigl\{
|M|:\ M\text{ is stable under }q^r+u,\;
u\in\{0,1\}^{|\Hh|},\;
\sum_{h\in\Hh}u_h\leq B
\Bigr\}.
\]
The cumulative increase in matching cardinality is
\(\Delta_r(B)=m_r(B)-m_r(0)\), and the marginal increase is
\(d_r(B)=m_r(B)-m_r(B-1)\) for \(B\geq1\).

We retain the 9,810 markets that admit a stable
matching at their original capacities, as established in
Section~\ref{subsec:existence-comparison}.
These markets cover all twelve couple count groups.
We reuse the exact values at \(B=0\) and solve \ICMCUTUC{}
for \(B=1,\ldots,10\), giving 98,100 new optimization runs.
Each run has a time limit of 600 seconds.

The solver proves optimality in 97,999 runs, while 101 reach the
time limit. We tighten the bounds using the fact that maximum matching
cardinality cannot decrease as the budget increases.
When all doctors belong to couples, we also use the fact
that the number of matched doctors must be even. This certifies eight additional
optima, leaving 93 values of \(m_r(B)\) known only within certified
bounds. Mean and median running times are 7.61 and 0.60 seconds.

\begin{figure}
	\centering
	\includegraphics[width=0.90\textwidth]
	{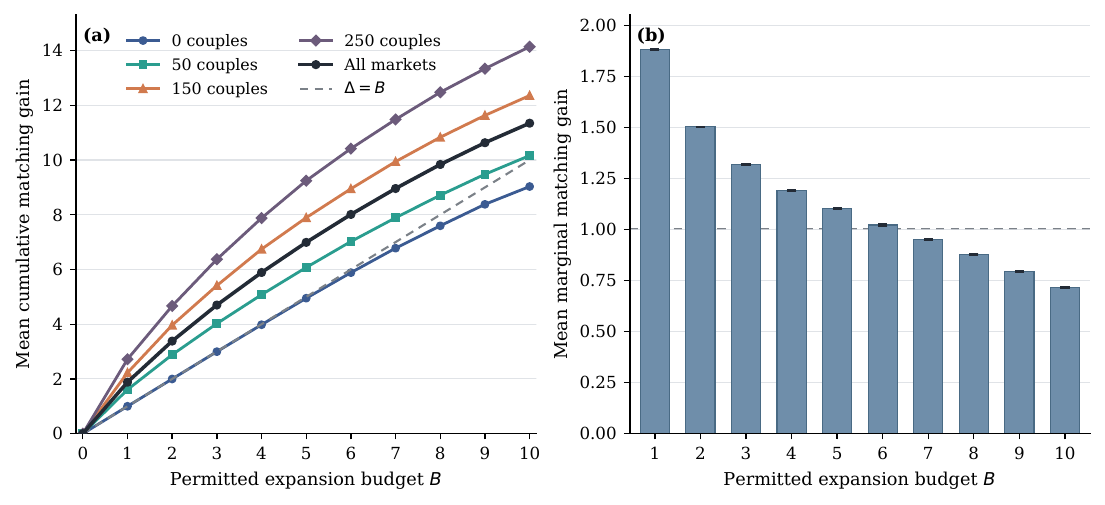}
	\caption{Certified matching gains under unit capacity expansion in the
		9,810 markets that admit a stable matching at the original capacities.}
	\label{fig:icm-uc-cardinality}
\end{figure}

Figure~\ref{fig:icm-uc-cardinality} reports mean cumulative
increases for four couple count groups and the full cohort in
panel~\textnormal{(a)}, and the pooled mean marginal increase
in panel~\textnormal{(b)}.
The pooled means include all 9,810 markets across the twelve
groups. Curves and bars show interval midpoints, while bands
and error bars indicate certified bounds.
The dashed lines mark \(\Delta=B\) and \(d=1\), respectively,
with \(B\) denoting the permitted budget rather than actual
additions.

The pooled mean marginal increase declines from 1.881--1.882
doctors for the first budget unit to 0.711--0.719 for the tenth.
At \(B=10\), the mean cumulative increase reaches
11.342--11.348 doctors.
Already at \(B=1\), 63.9\% of markets gain more than one
matched doctor. At \(B=10\), the mean cumulative increase is
9.031 doctors in markets without couples and
14.096--14.182 in markets with 250 couples.

Matching cardinality can increase by more than \(B\) because
expansion may support a different stable matching that also
fills previously unused positions at other hospitals.
The results show diminishing average marginal benefits over
the tested budget range. Choosing an optimal expansion budget
requires weighing these benefits against capacity costs,
which are not specified in this experiment.

\subsection{Scalability of the Exact and Cutoff formulations}
\label{subsec:scalability-experiment}

This experiment compares the scalability of Exact and Cutoff
for fixed capacity optimization and, when necessary, minimum
capacity repair. Balanced modification and unit expansion
are evaluated separately.

We independently generate fifty markets for each size
\(n\in\{100,\allowbreak 500,\allowbreak 1000,\allowbreak
2500,\allowbreak 5000,\allowbreak 10000,\allowbreak
15000,\allowbreak 20000,\allowbreak 25000,\allowbreak 30000\}\),
giving 500 markets in total.
Each market contains \(n\) doctors, \(0.10n\) hospitals, and
\(\lfloor0.05n/2\rfloor\) randomly formed couples,
corresponding to a target share of 5\% of doctors in couples.
Each doctor independently draws a uniformly random ordered
list of ten distinct hospitals.
A shared symmetric compatibility matrix filters the Cartesian
product of each couple's individual lists.
Each draw makes an unordered pair of distinct hospitals
compatible independently with probability 0.75, while every
hospital is compatible with itself.
The retained options follow the rank sum rule and tie breaking
described in Section~\ref{subsec:n100-experiments}.

At each size, 25 markets use independently generated hospital
priorities and 25 use a common master list.
Within each group, twenty markets have total capacity \(Q=n\)
and five have \(Q=0.95n\).
Each hospital first receives one position.
The remaining \(Q-|\Hh|\) positions are allocated independently
and uniformly across hospitals.

For each formulation, the fixed capacity model maximizes
matching cardinality. Repair is attempted only if that model
proves infeasibility.
BC minimizes the number of transfers, and UC minimizes the
number of added positions.
Neither repair model uses matching cardinality as a secondary
objective.
Each formulation and capacity policy receives 600 seconds per
market for the fixed capacity stage and any required repair.
BC and UC each receive the time remaining after their
formulation's fixed capacity stage, independently of the time
used by the other repair model.
A procedure is counted as complete if it certifies a maximum
cardinality stable matching at the original capacities or,
after proving infeasibility there, certifies a minimum repair
or proves that no admissible repair exists.

\begin{figure}[H]
  \centering
  \includegraphics[width=0.90\textwidth]{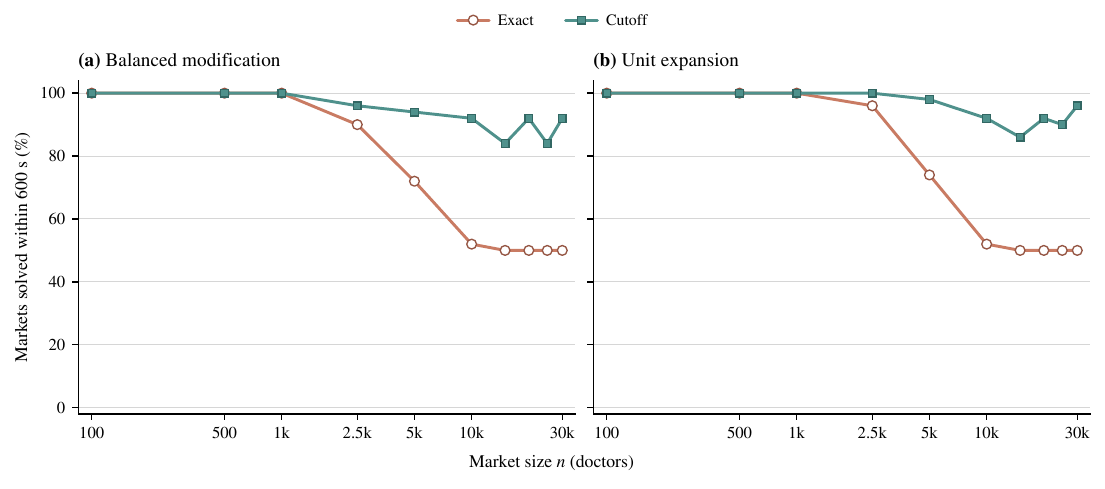}
  \caption{Percentage of markets for which Exact and Cutoff complete fixed
  capacity optimization and any required repair within a shared 600-second
  limit, under balanced modification (left) and unit expansion (right).}
  \label{fig:scalability-solution-rates}
\end{figure}

Figure~\ref{fig:scalability-solution-rates} shows that all four
procedures achieve a completion rate of 100\% for \(n\leq1000\).
At \(n=30000\), Cutoff completes 92\% of cases under balanced
modification and 96\% under unit expansion, compared with
50\% under either policy for Exact.
Across all 500 markets, the corresponding completion rates
are 93.4\% and 95.4\% for Cutoff, and 71.4\% and 72.2\%
for Exact.

The advantage of Cutoff arises primarily in the fixed capacity
stage. It certifies maximum cardinality stable matchings in
458 markets and proves infeasibility in the remaining 42.
Exact certifies maximum cardinality in 345 markets and proves
infeasibility in eighteen, leaving 137 unresolved.

Among the 42 markets entering repair under Cutoff, minimum
repairs are certified in nine under BC and nineteen under UC.
For Exact, the corresponding counts are twelve and sixteen
among eighteen markets.
All other repair runs reach the time limit.
Every certified minimum balanced repair uses one transfer,
while the certified minimum expansion budgets are one or two.
These repair results are not directly comparable because
the formulations enter this stage on different market sets
and with different remaining time limits.
The completion rates therefore compare the complete
solution procedures.

\section{Conclusion}
\label{sec:conclusion}

This paper shows how limited changes in hospital capacities can restore
stability in markets with couples. For arbitrary strict preferences in our
model, which excludes partially matched couples, a BIS-stable matching always
exists when each hospital can gain or lose at most two positions and total
capacity is preserved or increases by one. This net capacity bound is tight and
allows additions at some hospitals to be offset by reductions elsewhere.
Under sub-responsive and sub-complete couple preferences, reductions can be
avoided: at most one additional position per hospital and no more positions
added in total than there are couples suffice, with a tight aggregate bound.
Balanced unit modification, however, can fail even within this restricted
preference domain.

The equivalent Exact and Cutoff formulations allow the organizer to test
whether an admissible capacity modification restores stability and to
minimize its size. Their relative performance depends on the task.
Exact is modestly faster in the fixed capacity variants derived from the
archive, whereas Cutoff is more effective in the paired minimum repair experiments. Cutoff
also certifies more outcomes in the scalability experiment with up to
30,000 doctors and 600 seconds for the
fixed capacity and repair stages together. Much of this advantage arises
before capacity modification is attempted, and solving the minimum repair
problem remains difficult at larger scales.

The capacity experiments suggest that small modifications are often
sufficient in the tested markets. Among the 96,000 synthetic and reconstructed
markets, all 6,087 fixed capacity failures with initial total capacity equal
to the number of doctors are repaired by a single balanced transfer.
This regularity suggests considering a small
reallocation before expansion in comparable markets, but provides no general
existence guarantee. Expansion can also increase maximum
matching cardinality in initially stable markets, although average marginal
improvements decline over the tested budgets. Because the experiments do not
model adjustment costs, they do not determine an economically optimal policy
or expansion budget.

It remains open whether allowing each hospital to gain or lose at most one
position guarantees a BIS-stable matching under arbitrary strict preferences
in our model, without an aggregate capacity constraint. Whether expansion by
at most one position per hospital suffices is also open. Further directions
include identifying conditions that
guarantee balanced modification and developing efficient constructions for
the tight net capacity bound. Incorporating heterogeneous modification costs
and uncertainty about preferences would connect these existence and
optimization results more closely to capacity planning in practice.

\section*{Acknowledgements}
Zexiao Wu and Genjiu Xu acknowledge support from the National Key R\&D Program of China
(Grant No.~2021YFA1000402), the National Natural Science Foundation of China
(Grant No.~72571215), and the Guangdong Basic and Applied Basic Research
Foundation (Grant No.~2024A1515011241).

\appendix
\makeatletter
\renewcommand{\@seccntformat}[1]{%
  \ifstrequal{#1}{section}{\appendixname\space}{}%
  \csname the#1\endcsname.\hskip 0.5em}
\makeatother
\numberwithin{table}{section}
\renewcommand{\theHtable}{\Alph{section}.\arabic{table}}

\section{Alternative stability conditions}
\label{app:alternative-stability}

\subsection{Definitions and relationships}
\label{app:alternative-stability-definitions}

KPR-stability, MM-stability, and \NV{} retain S1 and S2 from
Definition~\ref{def:bis-blocking}; only the condition
in S3 changes. Let \(M\) be a feasible matching and consider a couple
\(k=(a,b)\) and an acceptable option \((h,h)\in\Pk_k\), where the couple
is unmatched or strictly prefers \((h,h)\) to \((M(a),M(b))\).
For a set of doctors \(A\), let \(\operatorname{Ch}^{q}_h(A)\) denote
the doctors in \(A\) who are acceptable to \(h\) and have the highest
priority at \(h\), up to \(q_h\) doctors.

\emph{KPR-stability.} Both members must be retained when the hospital
chooses from its incumbents and the two applicants:
\[
  \{a,b\}\subseteq
  \operatorname{Ch}^{q}_h\bigl(M(h)\cup\{a,b\}\bigr).
\]

\emph{MM-stability.} Any proposing member already assigned to \(h\)
retains that position. Each remaining member must be accommodated in a
distinct position that is either vacant or occupied by a doctor outside
\(\{a,b\}\) whom that member outranks.

\emph{\NV{}-stability.} The hospital must select each member when that
member is considered separately:
\[
  a\in\operatorname{Ch}^{q}_h\bigl(M(h)\cup\{a\}\bigr)
  \qquad\text{and}\qquad
  b\in\operatorname{Ch}^{q}_h\bigl(M(h)\cup\{b\}\bigr).
\]
Unlike MM, the two individual acceptance conditions may rely on
the same vacant position or the same displaced doctor. Under strict
hospital priorities, the condition in
\citet[Definition~2]{nguyen2018} coincides with the KPR condition.
\citet[Online Appendix~A.4]{nguyen2018} further show that the
matching constructed for their Theorem~1 satisfies the stronger
requirement that the hospital would select each member when that
member is considered separately alongside its current assignees.
We refer to this requirement as \NV{}-stability.

For a fixed instance, let \(\mathcal S(X)\) be the set of
\(X\)-stable matchings. Then
\[
  \mathcal S(\NV)\subseteq
  \mathcal S(\mathrm{BIS})\cap\mathcal S(\mathrm{MM}),
  \qquad
  \mathcal S(\mathrm{BIS})\cup\mathcal S(\mathrm{MM})
  \subseteq\mathcal S(\mathrm{KPR}).
\]
The BIS, MM, and KPR relationships are established in
\citet[Section~3.2]{delorme2021}. Figure~\ref{fig:stability-relationships} and the examples below illustrate
the relationships among all four stability notions.

\begin{figure}[H]
  \centering
  \includegraphics[width=0.50\linewidth]{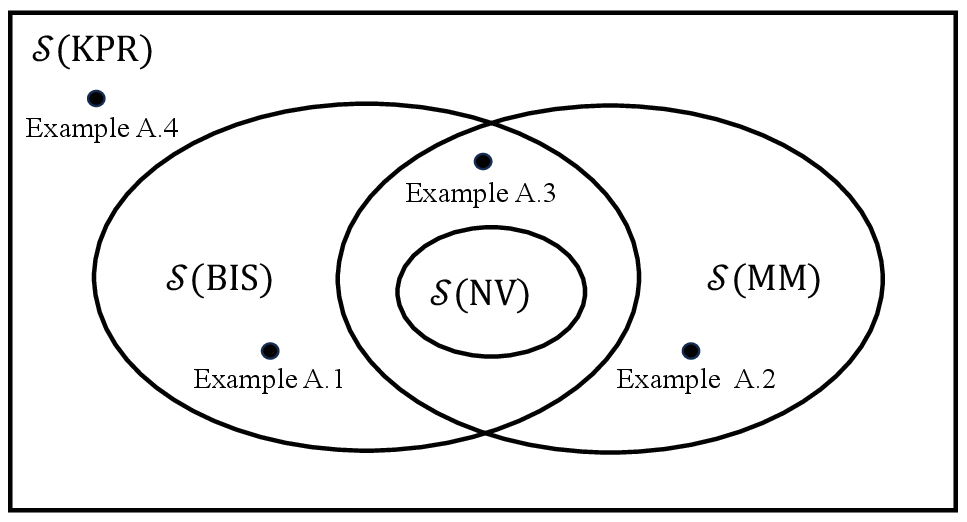}
  \caption{Relationships among the four stability concepts.}
  \label{fig:stability-relationships}
\end{figure}

Consider one hospital \(h\) with \(q_h=2\).
In all four examples, \(a,b\) are single doctors, and \((d_1,d_2)\)
and \((d_3,d_4)\) are couples. The singles find only \(h\) acceptable,
and each couple lists only \((h,h)\).
Table~\ref{tab:stability-relationship-examples} gives \(M(h)\)
and the hospital's priority order.

\begin{table}[H]
  \centering
  \caption{Four examples for Figure~\ref{fig:stability-relationships}.}
  \label{tab:stability-relationship-examples}
  \begin{tabular}{@{}cllcccc@{}}
    \toprule
    Example & \(M(h)\) & Hospital priority & KPR & BIS & MM & \NV{} \\
    \midrule
    \refstepcounter{example}\label{ex:stability-bis-only}\theexample
      & \(\{a,b\}\) & \(d_1\succ_h a\succ_h d_2\succ_h b\succ_h d_3\succ_h d_4\)
      & \(\checkmark\) & \(\checkmark\) & \(\times\) & \(\times\) \\
    \refstepcounter{example}\label{ex:stability-mm-only}\theexample
      & \(\{d_3,d_4\}\) & \(d_3\succ_h d_1\succ_h d_2\succ_h d_4\succ_h a\succ_h b\)
      & \(\checkmark\) & \(\times\) & \(\checkmark\) & \(\times\) \\
    \refstepcounter{example}\label{ex:stability-bis-mm}\theexample
      & \(\{a,b\}\) & \(a\succ_h d_1\succ_h d_2\succ_h b\succ_h d_3\succ_h d_4\)
      & \(\checkmark\) & \(\checkmark\) & \(\checkmark\) & \(\times\) \\
    \refstepcounter{example}\label{ex:stability-kpr-only}\theexample
      & \(\{d_3,d_4\}\) & \(d_1\succ_h d_3\succ_h d_2\succ_h d_4\succ_h a\succ_h b\)
      & \(\checkmark\) & \(\times\) & \(\times\) & \(\times\) \\
    \bottomrule
  \end{tabular}
\end{table}

\subsection{Exact and Cutoff formulations}
\label{app:alternative-S3-cutoffs}

Only the S3 constraints change across stability notions.
For every \(k\in\C\) and \(p=(h,h)\in\Pk_k\), relabel the members so that
\(i_1\succ_h i_2\). For MM, let
\(\mathcal R_h(i_1,i_2):=\{j:(j,h)\in\E,\ i_1\succeq_h j\succ_h i_2\}\).
This priority interval contains \(i_1\) and excludes \(i_2\), so it is
nonempty. Replace the S3 constraint by the appropriate constraint below:
\begin{align}
  q_h\left(1-\sum_{p'\in\Pk_p^{\succeq}(k)}y_{kp'}\right)
  -1+x_{i_1h}
  &\leq
  \sum_{j\in\D_{i_2}^{\succeq}(h)}x_{jh}-x_{i_2h},
  &&\textnormal{(KPR)},
  \label{eq:alternative-KPR-fixed-S3}\\*
  q_h\left(1-\sum_{p'\in\Pk_p^{\succeq}(k)}y_{kp'}\right)-1
  +\frac{\sum_{j\in\mathcal R_h(i_1,i_2)}x_{jh}}
  {|\mathcal R_h(i_1,i_2)|}
  &\leq
  \sum_{j\in\D_{i_2}^{\succeq}(h)}x_{jh}-x_{i_2h},
  &&\textnormal{(MM)},
  \label{eq:alternative-MM-fixed-S3}\\*
  q_h\left(1-\sum_{p'\in\Pk_p^{\succeq}(k)}y_{kp'}\right)
  &\leq
  \sum_{j\in\D_{i_2}^{\succeq}(h)}x_{jh}-x_{i_2h},
  &&\textnormal{(\NV{})}.
  \label{eq:alternative-NV-fixed-S3}
\end{align}
At fixed capacities, the Exact formulations for KPR, MM, and \NV{} are
obtained from IFC~\eqref{eq:fixed-model} by replacing
\eqref{eq:fixed-S3} with \eqref{eq:alternative-KPR-fixed-S3},
\eqref{eq:alternative-MM-fixed-S3}, or
\eqref{eq:alternative-NV-fixed-S3}, respectively. All other constraints and
the objective are unchanged.
The KPR and MM constraints are the strict preference specializations of
constraints~(17) and~(20), respectively, in
\citet[Section~4]{delorme2021}. Proposition~\ref{prop:alternative-NV-exactness}
proves the exactness of the \NV{} constraint.

Directly replacing
\(q_h\) by \(q'_h\) in the preceding constraints creates products of \(q'_h\) and
the couple preference prefix. Using \(0\leq q'_h\leq U_h\), the bounded
linearization in Section~\ref{subsec:variable-ip} gives
\begin{align}
  q'_h-U_h\sum_{p'\in\Pk_p^{\succeq}(k)}y_{kp'}
  -1+x_{i_1h}
  &\leq
  \sum_{j\in\D_{i_2}^{\succeq}(h)}x_{jh}-x_{i_2h},
  &&\textnormal{(KPR)},
  \label{eq:alternative-KPR-variable-S3}\\*
  q'_h-U_h\sum_{p'\in\Pk_p^{\succeq}(k)}y_{kp'}-1
  +\frac{\sum_{j\in\mathcal R_h(i_1,i_2)}x_{jh}}
  {|\mathcal R_h(i_1,i_2)|}
  &\leq
  \sum_{j\in\D_{i_2}^{\succeq}(h)}x_{jh}-x_{i_2h},
  &&\textnormal{(MM)},
  \label{eq:alternative-MM-variable-S3}\\*
  q'_h-U_h\sum_{p'\in\Pk_p^{\succeq}(k)}y_{kp'}
  &\leq
  \sum_{j\in\D_{i_2}^{\succeq}(h)}x_{jh}-x_{i_2h},
  &&\textnormal{(\NV{})}.
  \label{eq:alternative-NV-variable-S3}
\end{align}
At bounded variable capacities, the Exact formulations for KPR, MM, and
\NV{} are obtained from the linearized ICM by replacing
\eqref{eq:var-S3} with \eqref{eq:alternative-KPR-variable-S3},
\eqref{eq:alternative-MM-variable-S3}, or
\eqref{eq:alternative-NV-variable-S3}, respectively. All other constraints
and variable domains are unchanged. 

For MM and \NV{}, consider every option \(p_k^t=(h,h)\), relabel its
members so that \(i_1\succ_h i_2\), and impose, respectively,
\begin{align}
  \bigl[\min\{U_h,r_h(i_2)-r_h(i_1)\}-1\bigr]
  w_{h,r_h(i_2)-1}+w_{h,r_h(i_1)-1}
  &\leq
  \min\{U_h,r_h(i_2)-r_h(i_1)\}
  \bigl[1+(U_h-1)a_{kt}\bigr],
  \label{eq:cutoff-MM-S3}\\*
  w_{h,r_h(i_2)-1}&\leq U_ha_{kt}.
  \label{eq:cutoff-NV-S3}
\end{align}
All three Cutoff IPs retain
\eqref{eq:doctor}--\eqref{eq:link2}, \eqref{eq:fixed-S2c},
\eqref{eq:variable-capacity},
\eqref{eq:cutoff-assignee}--\eqref{eq:cutoff-S2b}, and the prefix and tail
recursions in Section~\ref{subsec:cutoff-formulations}. KPR additionally uses
\eqref{eq:cutoff-S3-base}; MM uses \eqref{eq:cutoff-MM-S3}; and \NV{} uses
\eqref{eq:cutoff-NV-S3}. 

\subsection{Proofs for \NV{} and Cutoff}
\label{app:alternative-formulation-proofs}

\begin{proposition}
\label{prop:alternative-NV-exactness}
Replacing \eqref{eq:fixed-S3} in IFC~\eqref{eq:fixed-model} by
\eqref{eq:alternative-NV-fixed-S3} yields feasible integer solutions that induce
exactly the \NV{}-stable matchings under \(q\). Under
\(0\leq q'_h\leq U_h\), replacing \eqref{eq:var-S3} by
\eqref{eq:alternative-NV-variable-S3} in the linearized formulation obtained
from ICM~\eqref{eq:variable-model} yields feasible integer solutions that induce
exactly the pairs \((M,q')\) in which \(M\) is \NV{}-stable under \(q'\).
Imposing
\eqref{eq:balanced-capacity-region},
\eqref{eq:unit-expansion-capacity-region}, or
\eqref{eq:near-balanced-capacity-region} restricts this correspondence to
the capacity vectors permitted by the selected policy.
\end{proposition}

\begin{proof}
The S1 and S2 conditions are unchanged under \NV{}. They are excluded by
\eqref{eq:fixed-S1}--\eqref{eq:fixed-S2c} under \(q\), and by
\eqref{eq:var-S1}--\eqref{eq:var-S2b} together with
\eqref{eq:fixed-S2c} under \(q'\). It remains to consider an option that
assigns both members of a couple to the same hospital.

Fix \(k\in\C\) and \(p_k^t=(h,h)\), and relabel the members as \(i_1,i_2\)
so that \(i_1\succ_h i_2\). Strict hospital priorities give, for either
member \(i\),
\[
  \sum_{j\in\D_i^{\succeq}(h)}x_{jh}-x_{ih}
  =\bigl|\{j\in M(h):j\succ_h i\}\bigr|.
\]
The subtraction of \(x_{ih}\) makes the equality valid also when
\(i\in M(h)\). Hospital \(h\) selects \(i\) from \(M(h)\cup\{i\}\) exactly
when this count is smaller than \(q'_h\). Since \(i_1\succ_h i_2\), the set
on the right for \(i_1\) is contained in that for \(i_2\). Thus \(h\)
selects both members when each is considered separately exactly when
\(\sum_{j\in\D_{i_2}^{\succeq}(h)}x_{jh}-x_{i_2h}<q'_h\).

At an integer solution,
\(\sum_{p'\in\Pk_{p_k^t}^{\succeq}(k)}y_{kp'}\) is zero exactly when
couple \(k\) is unmatched or receives an option worse than \(p_k^t\).
Hence \(p_k^t\) is an \NV{} block exactly when this prefix sum is zero and
the preceding strict inequality holds. Excluding the block is therefore
equivalent to
\[
  q'_h\left(
    1-\sum_{p'\in\Pk_{p_k^t}^{\succeq}(k)}y_{kp'}
  \right)
  \leq
  \sum_{j\in\D_{i_2}^{\succeq}(h)}x_{jh}-x_{i_2h}.
\]
Setting \(q'_h=q_h\) gives
\eqref{eq:alternative-NV-fixed-S3}.

Under \(0\leq q'_h\leq U_h\), the displayed inequality is equivalent to
\eqref{eq:alternative-NV-variable-S3} at every integer solution. If the
prefix sum is zero, the two inequalities coincide. If it is one, their
left-hand sides are zero and \(q'_h-U_h\leq0\), respectively, while the
right-hand side is nonnegative; both are therefore redundant. The case
\(q'_h=0\) is included, since the hospital then selects neither member.

Applying this argument to every \(p_k^t=(h,h)\) excludes exactly the
\NV{} S3 blocks. Together with the S1 and S2 constraints identified above,
this proves both statements.
\end{proof}

\begin{proposition}
\label{prop:alternative-cutoff-equivalence}
For each of KPR-stability, MM-stability, and \NV{}, the Exact and Cutoff
formulations in Appendix~\ref{app:alternative-S3-cutoffs} admit exactly
the same matchings and capacity vectors.
\end{proposition}

\begin{proof}
The S1 and S2 constraints do not depend on the choice among KPR, MM, and
\NV{}. At fixed capacities they are
\eqref{eq:fixed-S1}--\eqref{eq:fixed-S2c}; with variable capacities they are
\eqref{eq:var-S1}--\eqref{eq:var-S2b} together with
\eqref{eq:fixed-S2c}. The corresponding Cutoff constraints are
\eqref{eq:cutoff-assignee}--\eqref{eq:cutoff-S2b}, together with the prefix
recursions and \eqref{eq:fixed-S2c}. The S1 and S2 arguments in the proof of
Theorem~\ref{thm:cutoff-equivalence} show that these constraints exclude
exactly the same blocks. It remains to compare the S3 constraints.

Fix the binary assignment variables and an admissible integer capacity vector
\(q'\). Consider \(p_k^t=(h,h)\), and relabel the members so that
\(i_1\succ_h i_2\). For this proof, write
\[
  A:=a_{kt},\qquad
  W_\nu:=w_{h,r_h(i_\nu)-1}\quad(\nu=1,2),\qquad
  d:=r_h(i_2)-r_h(i_1),\qquad
  \mu:=\min\{U_h,d\}.
\]
The prefix and tail recursions imply
\[
  A=\sum_{p'\in\Pk_{p_k^t}^{\succeq}(k)}y_{kp'}\in\{0,1\},
  \qquad
  W_\nu=q'_h-
  \left(\sum_{j\in\D_{i_\nu}^{\succeq}(h)}x_{jh}-x_{i_\nu h}\right),
\]
and
\[
  0\leq W_2\leq W_1\leq q'_h\leq U_h,
  \qquad
  W_1-W_2=\sum_{j\in\mathcal R_h(i_1,i_2)}x_{jh}\leq\mu.
\]
Here \(d=|\mathcal R_h(i_1,i_2)|\geq1\), by strict hospital priorities.
Moreover, \(A=0\) exactly when the couple is unmatched or receives an option
worse than \(p_k^t\). Using the preceding identities, the Exact S3
constraints for KPR, MM, and \NV{} are equivalent, respectively, to
\[
\begin{array}{rcl}
  W_2+x_{i_1h}
  &\leq& 1+U_hA,\\
  (d-1)W_2+W_1
  &\leq& d(1+U_hA),\\
  W_2
  &\leq& U_hA.
\end{array}
\]
The corresponding Cutoff constraints are
\[
\begin{array}{rcl}
  W_2+x_{i_1h}
  &\leq& 1+(U_h-1)A,\\
  (\mu-1)W_2+W_1
  &\leq& \mu\bigl[1+(U_h-1)A\bigr],\\
  W_2
  &\leq& U_hA.
\end{array}
\]

Suppose first that \(A=0\). The KPR inequalities coincide, as do the \NV{}
inequalities. For MM, the two inequalities reduce to
\[
  (d-1)W_2+W_1\leq d
  \qquad\text{and}\qquad
  (\mu-1)W_2+W_1\leq\mu.
\]
If \(\mu=0\), both hold because \(U_h=W_1=W_2=0\). Suppose
\(\mu\geq1\). For either \(c=d\) or \(c=\mu\), the preceding bounds give
\(0\leq W_2\leq W_1\) and \(W_1-W_2\leq c\). Since \(W_1\) and \(W_2\)
are integers,
\[
  (c-1)W_2+W_1\leq c
  \quad\Longleftrightarrow\quad
  W_2=0\ \text{or}\ W_1\leq1.
\]
Indeed, conditional on \(W_2=0\), the difference bound gives
\(W_1\leq c\); if \(W_1\leq1\), the inequality also holds. Conversely,
\(W_2\geq1\) and \(W_1\geq2\) make the left side at least \(c+1\).
Applying this equivalence with \(c=d\) and \(c=\mu\) proves that the Exact
and Cutoff MM constraints impose the same condition. This condition is the
negation of the MM requirement that \(h\) provide one admissible position for
\(i_2\) and a second, distinct admissible position for \(i_1\).

Now suppose that \(A=1\). All six inequalities are automatic. Because
\(i_1\in\mathcal R_h(i_1,i_2)\),
\[
  W_2+x_{i_1h}
  \leq W_2+\sum_{j\in\mathcal R_h(i_1,i_2)}x_{jh}
  =W_1\leq U_h.
\]
This proves the KPR Cutoff inequality, and the Exact inequality is weaker.
If \(\mu\geq1\), then
\[
  (\mu-1)W_2+W_1\leq\mu U_h,
\]
which proves the MM Cutoff inequality; the Exact MM inequality follows from
\((d-1)W_2+W_1\leq dU_h\). If \(\mu=0\), then \(U_h=W_1=W_2=0\).
Finally, both \NV{} inequalities reduce to \(W_2\leq U_h\).

Thus the Exact and Cutoff S3 constraints agree for every same hospital option.
For fixed capacities, take \(q'_h=U_h=q_h\). For any bounded variable capacity
policy, the argument applies to every admissible \(q'\). Together with the S1
and S2 equivalence and the identical assignment and capacity constraints, this
proves that the two formulations admit the same matchings and capacity vectors.
\end{proof}

\section{Proofs}
\label{app:structured-existence}

\subsection{Proof of Theorem~\ref{thm:general-capacity-bound}}
\label{app:general-capacity-bound-proof}

We modify the iterative rounding procedure of \citet{nguyen2018}. We retain
their dominating fractional assignment and constraint deletion rule, but keep
every active hospital constraint as an equality and change the treatment of
slack assignment constraints when the aggregate capacity constraint is
removed. The modification uses our assumption that a couple is either
assigned jointly or left entirely unmatched: every acceptable couple
assignment uses exactly two hospital positions. We first recall the
fractional construction and then prove that the modified iteration terminates
with an increase in total capacity of at most one.

\paragraph{The fractional construction.}
Treat each single doctor and each couple as one agent and write
\(\mathcal R:=\Sg\cup\C\). For \(r\in\mathcal R\), let \(\Omega_r\) be its
set of acceptable assignment columns and set
\(\Omega:=\bigcup_{r\in\mathcal R}\Omega_r\). Let \(A_{r\omega}=1\) if
\(\omega\in\Omega_r\) and zero otherwise, so \(A_rx\leq1\) is the
assignment constraint for \(r\). For each \(h\in\Hh\) and
\(\omega\in\Omega\), let \(H_{h\omega}\in\{0,1,2\}\) be the number of doctors
assigned to \(h\) by \(\omega\), and write
\(H_h=(H_{h\omega})_{\omega\in\Omega}\). The fractional assignment polytope is
\[
 P(q):=\{x\geq0:A_rx\leq1\ (r\in\mathcal R),\quad
 H_hx\leq q_h\ (h\in\Hh)\}.
\]

Order the columns in each agent row by that agent's reported preferences. For
a column \(\omega\) incident to hospital \(h\), let \(d_h(\omega)\) be the
doctor assigned to \(h\) by \(\omega\); if \(\omega\) assigns both members of
a couple to \(h\), let \(d_h(\omega)\) be the lower-priority member at \(h\).
Following \citet[Online Appendix~A3]{nguyen2018}, hospital \(h\) orders its
incident columns by \(d_h(\omega)\) according to \(\succ_h\), breaking a tie
between columns of the same couple by that couple's preference. This is a
technical row order for the domination argument, not a change in the
hospital's priority over doctors.

With these row orders, Scarf's lemma \citep{scarf1967} gives a dominating
extreme point \(\widetilde x\) of \(P(q)\). Let
\(H^0:=\{h\in\Hh:H_h\widetilde x<q_h\}\). The dummy doctor construction of
\citet[Online Appendix~B2]{nguyen2018} adds single doctors who find only the
hospitals in \(H^0\) acceptable and extends \(\widetilde x\) to a dominating
vector that fills every hospital. Taking a vertex of the minimal face of the
augmented polytope containing this vector preserves its support and all tight
rows. We may therefore denote this vertex by \(\bar x\) and take it to be
extreme, dominating, and such that \(H_h\bar x=q_h\) for every hospital. From
this point onward, the agent set and assignment columns include the dummy
single doctors.

Lemma~3 and Online Appendix~B1 of \citet{nguyen2018} give the conditions that
preserve domination during rounding. It suffices to obtain an integral vector
\(x^*\) that respects the individual assignment constraints and satisfies
\[
 \bar x_\omega=0\ \Longrightarrow\ x^*_\omega=0,
 \qquad
 A_r\bar x=1\ \Longrightarrow\ A_rx^*=1.
\]
With \(q'_h:=H_hx^*\), these conditions make \(x^*\) dominating in the
augmented market. We verify BIS-stability directly below. Our task is to
preserve these conditions while obtaining the stronger aggregate bound.

Set \(w_r=1\) for a single doctor, including a dummy doctor, and \(w_r=2\)
for a couple, and define
\(a_\omega=w_r\) for \(\omega\in\Omega_r\). Every single column uses one
hospital position. Because partial couple assignments are excluded, every
couple column uses two, including a column that assigns both members to the
same hospital. Thus the total capacity row has the identity
\begin{equation}
 a=\sum_{h\in\Hh}H_h=\sum_{r\in\mathcal R}w_rA_r.
 \label{eq:connected-total-row}
\end{equation}
This identity is the restriction on which the improvement relies; it need
not hold when a couple may accept an assignment of only one member.
Let \(K:=\sum_{h\in\Hh}q_h\).
Since every such hospital is full at \(\bar x\), \(a\bar x=K\).

\paragraph{The modified iteration.}
Start from \(\bar x\). Once a variable becomes zero or one, keep it fixed.
Once an agent constraint becomes tight, retain it as an equality. Let
\(\mathcal B:=\{r:A_r\bar x=1\}\) initially, and add \(r\) to
\(\mathcal B\) whenever \(A_rx=1\) at a later extreme point. Thus
\(A_rx=1\) is retained for \(r\in\mathcal B\), while the other agent
constraints remain inequalities unless they are replaced below. This rule
applies only to individual constraints still present in the system; a
constraint replaced by the grouped row below is not reintroduced separately.

Let \(\mathcal U=\Hh\) initially. Retain \(H_hx=q_h\) for
\(h\in\mathcal U\) until that hospital constraint is removed, and retain
\(ax\leq K\) until the aggregate constraint is removed. After any removal or
replacement, maximize \(ax\) over the resulting system and choose an extreme
maximizer. The preceding point remains feasible after every change. Fix any
newly integral variables and retain any newly tight agent constraint as an
equality. Consequently, \(ax\) cannot decrease. Since \(a\bar x=K\), every
current point satisfies \(ax=K\) while the aggregate constraint remains.

For a current point \(x\), let
\(J(x):=\{\omega\in\Omega:0<x_\omega<1\}\) and
\(\rho_h(x):=\sum_{\omega\in J(x)}H_{h\omega}\). If \(x\) is integral,
stop. If some \(h\in\mathcal U\) satisfies \(\rho_h(x)\leq3\), remove its
capacity equality.

Suppose no hospital constraint can be removed and the aggregate constraint is
still present. Let
\(F:=\{r\in\mathcal R\setminus\mathcal B:(A_r)_{J(x)}\neq0\}\). These agent
constraints are slack. The token count in
\citet[Online Appendix~B3, Case~1]{nguyen2018}, which remains valid when the
active hospital constraints are kept as equalities, gives \(|F|\leq2\). The
count uses only that the retained hospital rows are tight and that none
satisfies \(\rho_h(x)\leq3\), both of which hold here.
Put \(y_r:=A_rx\) for \(r\in F\) and
\(Y:=\sum_{r\in F}y_r\).

If \(F=\varnothing\), remove the aggregate constraint. If \(Y>1\), remove it
and retain \(A_rx\leq1\) for \(r\in F\). If \(F\neq\varnothing\) and
\(Y\leq1\), replace these individual constraints by
\(Gx:=\sum_{r\in F}A_rx\leq1\) and remove the aggregate constraint. The
current point satisfies the new constraint because \(Gx=Y\leq1\). Moreover,
nonnegativity gives
\(A_ru\leq Gu\leq1\) for every \(r\in F\) and every feasible vector \(u\),
so the replacement preserves individual assignment feasibility. Retain
\(Gx\leq1\) for the rest of the procedure. Unlike deletion alone, this
replacement may restrict the feasible region, but the preceding point
remains feasible. This is all that is needed for \(ax\) to remain
nondecreasing.

\paragraph{Termination.}
The only new issue is termination after the grouped row has been introduced.
Before the aggregate constraint is removed, progress follows from the
Case~1 token count above. After its removal, the argument of
\citet[Online Appendix~B3, Case~2]{nguyen2018} applies if no constraint
\(Gx\leq1\) was introduced. It also applies when \(G\) was introduced but
\(G_J=0\), where \(J:=J(x)\) and a subscript \(J\) denotes restriction to
the columns in \(J\). Indeed, \(G=\sum_{r\in F}A_r\) and all entries are
nonnegative, so \(G_J=0\) implies \((A_r)_J=0\) for every \(r\in F\).
After restriction to \(J\) and omission of these zero rows, the remaining
fractional system is the one considered in their Case~2.

It remains to consider \(G_J\neq0\). Suppose, for a contradiction, that the
current point is fractional and that no hospital constraint can be removed.
Then \(\rho_h(x)\geq4\) for every \(h\in\mathcal U\). The rank argument in
\citet[Online Appendix~B3]{nguyen2018} gives a set \(\mathcal T\) of
\(|J|\) linearly independent tight rows after restriction to \(J\). These
are active hospital rows, tight agent rows, and, if it is tight, the row
\(G\).

Give each variable in \(J\) one token and use the distribution in
\citet[Online Appendix~B3]{nguyen2018}. A couple column gives
\(H_{h\omega}/4\) token to each hospital row \(h\) and the remaining
\(1/2\) token to its couple row. A single column gives \(1/4\) token to its
hospital row and \(3/4\) token to its single doctor row. For a column
belonging to an agent in \(F\), send the share intended for its agent row to
\(G\) instead. Each fractional column still distributes exactly one token.

Every active hospital row receives at least one token because
\(\rho_h(x)\geq4\). Every nonzero tight couple row receives at least one
token, since a tight agent row containing a fractional variable contains at
least two such variables. A nonzero tight single row receives at least
\(3/2\) tokens. If \(G\) is tight, then \(Gx=1\). The condition
\(G_J\neq0\) rules out a variable fixed at one in this row, since
nonnegativity would then leave no positive fractional variable in \(G\). At least two
fractional variables are therefore needed to make the entries in the row sum
to one. Thus \(G\) also receives at least one token.

If \(G\) is slack, the rows in \(\mathcal T\) receive at least \(|J|\)
tokens, while \(G_J\neq0\) sends an additional positive amount to \(G\).
This exceeds the total supply of \(|J|\) tokens. Suppose instead that \(G\)
is tight. All tokens must then be sent to rows in \(\mathcal T\), every such
row must receive exactly one, and every nonzero tight row must belong to
\(\mathcal T\). No token can be sent to a slack row or to a hospital row
that has been removed.

A fractional single column is impossible. If it does not belong to \(G\),
its own row either is slack and receives a token outside \(\mathcal T\), or
is tight and receives at least \(3/2\) tokens. If it belongs to \(G\), it
gives \(3/4\) token to \(G\); another fractional column is needed for
\(Gx=1\), so \(G\) receives at least \(3/4+1/2>1\) tokens. Hence every
column in \(J\) is a couple column.

Token exhaustion also implies that every hospital row used by such a column
is active and that its agent incidence is represented either by \(G\) or by a
tight couple row. Let
\(\mathcal C_J:=\{c\in\C\setminus F:(A_c)_J\neq0\}\). Since every couple
column uses two hospital positions, equality holds column by column in
\[
 \sum_{h\in\mathcal U}(H_h)_J
 =2\left(\sum_{c\in\mathcal C_J}(A_c)_J+G_J\right).
\]
Every nonzero row in this equality is tight and belongs to \(\mathcal T\).
The equality is therefore a nontrivial linear dependence among the rows of
\(\mathcal T\), contradicting their independence. Some hospital constraint
must consequently satisfy \(\rho_h(x)\leq3\) and can be removed.

Every fractional iteration therefore removes a hospital constraint or, once,
the aggregate constraint. Each hospital constraint is removed at most once,
so the procedure terminates at an integral vector \(x^*\).

\paragraph{The aggregate bound.}
If the procedure terminates while the aggregate constraint remains, then
\(ax^*=K\). Otherwise, let \(\widehat x\) be the point at which that
constraint is removed. With \(F\), \(y_r\), and \(Y\) defined at that point,
set \(L:=\sum_{r\in F}w_ry_r\) and
\(C:=\sum_{r\notin F}w_rA_r\widehat x\).
The contribution of every agent outside \(F\) is fixed from this point
onward. An agent in \(\mathcal B\) continues to contribute \(w_r\). An
agent outside \(\mathcal B\cup F\) has no fractional variable and has a
slack constraint, so all of its variables are fixed at zero. Thus \(C\) is
an integer and remains unchanged. Since \(a\widehat x=K\),
$K=C+L$.
Hence \(L\) is an integer. If \(F\neq\varnothing\), then
\(0<y_r<1\) for every \(r\in F\), so \(L\) is a positive integer.

Suppose \(Y>1\). Since \(|F|\leq2\) and each \(y_r<1\), \(F\) contains
exactly two agents. If both are couples, then \(L=2Y\in(2,4)\), so \(L=3\),
while their final contribution is at most four. If one is a couple and one
is a single doctor, then \(L=2y_c+y_s=Y+y_c\in(1,3)\), so \(L=2\), while
their final contribution is at most three. Two single doctors would give
\(L=Y\in(1,2)\), contradicting the integrality of \(L\). Thus the
contribution of \(F\) can increase by at most one.

If \(F\neq\varnothing\) and \(Y\leq1\), the combined constraint remains in
the system. Since \(w_r\leq2\), its final contribution satisfies
\(\sum_{r\in F}w_rA_rx^*\leq2Gx^*\leq2\). The contribution of \(F\) when
the aggregate constraint is removed is the positive integer \(L\), so
\(L\geq1\), and it can again increase by at most one. The case
\(F=\varnothing\) gives no increase. Therefore
\(ax^*\leq K+1\).

For the reverse inequality, the point \(\widehat x\) has objective value
\(K\) and remains feasible after the aggregate constraint is removed or
replaced. At every subsequent iteration, \(ax\) is maximized over a system
containing the preceding point, so it never decreases. Consequently,
\begin{equation}
 K\leq ax^*\leq K+1.
 \label{eq:net-one-rounded-total}
\end{equation}

\paragraph{Capacity changes and stability.}
Set \(q'_h:=H_hx^*\) in the augmented market. The hospital deletion estimate
is the one in the proof of \citet[Theorem~1, p.~3167]{nguyen2018}. If the
row for \(h\) is never removed, then \(q'_h=q_h\). If it is removed at a
point \(x\), then \(H_hx=q_h\), \(\rho_h(x)\leq3\), and
\(H_h\lfloor x\rfloor\leq q'_h\leq H_h\lceil x\rceil\).
The two integer deviations from \(q_h\) have sum at most three. If the row
contains a fractional variable, both are positive and hence at most two;
otherwise both are zero. Thus \(|q'_h-q_h|\leq2\) for every \(h\in\Hh\).

Every variable that is zero at \(\bar x\) remains zero, and every agent
constraint that is tight at \(\bar x\) remains tight. Individual assignment
feasibility also holds. The individual constraints for \(F\) remain in the
system when \(Y>1\); when they are replaced, nonnegativity gives
\(A_rx^*\leq Gx^*\leq1\) for every \(r\in F\). The support and tightness
argument in \citet[Online Appendix~B1]{nguyen2018} therefore shows that
\(x^*\) dominates every assignment column in the augmented market under the
capacity vector \(q'\).

Remove the dummy doctors from \(x^*\), but keep \(q'\) unchanged, and let
\(M\) be the resulting matching of the original doctors. No original column
was initially dominated at a hospital in
\(H^0\); its domination certificate was instead at its original agent row or
at a hospital outside \(H^0\). The support condition prevents a dummy doctor
from being assigned outside \(H^0\), so all such certificates survive the
removal of the dummy doctors, as in
\citet[Online Appendix~B2]{nguyen2018}. Hence every original assignment
column remains dominated under \(M\) and \(q'\).

We now verify BIS-stability directly. Suppose that \(M\) admits a blocking
coalition and let \(\omega\) be its assignment column. The proposing single
doctor or couple strictly prefers \(\omega\) to its assignment under \(M\),
so its agent row does not dominate \(\omega\). For an S1 block, the hospital
row is slack if the target hospital is undersubscribed. Otherwise, the
proposing doctor has higher priority than some incumbent \(j\). The selected
column containing \(j\) is represented at that hospital by \(j\), or by a
doctor with still lower priority when both members of an incumbent couple are
assigned there. It is therefore ranked below \(\omega\), so the hospital row
does not dominate \(\omega\).

For an S2 block, apply the same argument at each target hospital. If a
proposed member is already assigned to that hospital, the selected column and
\(\omega\) either have the same representative doctor, in which case the
couple's strict preference ranks \(\omega\) first, or the selected column
assigns both members there and has a weakly lower representative. Thus neither
target hospital row dominates \(\omega\).

Finally, consider an S3 block by a couple \((a,b)\) at hospital \(h\). In
cases \textnormal{(S3)(i)} and \textnormal{(S3)(ii)}, hospital \(h\) is
undersubscribed, so its row is slack. In case \textnormal{(S3)(iii)}, the
blocking condition provides an incumbent \(j\in M(h)\) below both \(a\) and
\(b\) in \(h\)'s priority order. The selected column containing \(j\) is
represented at \(h\) by \(j\), or by a still lower-priority doctor, whereas
\(\omega\) is represented by the lower-priority member of \(\{a,b\}\).
Consequently, \(\omega\) is ranked above that selected column and the hospital
row does not dominate it. In every case, \(\omega\) is dominated at none of
its incident rows, contradicting its surviving domination certificate.
Therefore \(M\) is BIS-stable under \(q'\).

Therefore, \eqref{eq:connected-total-row} and
\eqref{eq:net-one-rounded-total} give
\(0\leq\sum_{h\in\Hh}(q'_h-q_h)=ax^*-K\leq1\).
Together with the individual capacity
bound, this proves the theorem.

\subsection{Proof of Proposition~\ref{prop:a2b0-counterexample}}
\label{app:a2b0-counterexample}

Let \(\Hh=\{h_1,h_2,h_3\}\), \(q=(3,3,3)\), and let the doctors comprise
four couples \(k_1=(d_1,d_2)\), \(k_2=(d_3,d_4)\),
\(k_3=(d_5,d_6)\), \(k_4=(d_7,d_8)\), and two singles \(s_1,s_2\).
The preference lists and the common hospital priority order are as follows.

\begin{center}
\small
\setlength{\tabcolsep}{4pt}
\renewcommand{\arraystretch}{1.12}
\begin{tabular}{@{}ll@{}}
\toprule
Agent(s) & Preference list or priority order \\
\midrule
\(s_1,s_2\) & \(h_1\succ h_2\succ h_3\) \\
\(k_1=(d_1,d_2)\) & \((h_2,h_1)\succ(h_3,h_1)\) \\
\(k_2=(d_3,d_4)\) & \((h_3,h_1)\succ(h_2,h_1)\succ(h_2,h_3)\succ(h_3,h_2)\) \\
\(k_3=(d_5,d_6)\) & \((h_3,h_1)\succ(h_2,h_1)\succ(h_3,h_2)\) \\
\(k_4=(d_7,d_8)\) & \((h_1,h_3)\succ(h_2,h_3)\succ(h_2,h_1)\) \\
\midrule
\multicolumn{2}{@{}l@{}}{Common hospital priority order: \(d_2\succ d_6\succ d_7\succ d_3\succ d_8\succ s_1\succ d_1\succ s_2\succ d_4\succ d_5\)} \\
\bottomrule
\end{tabular}
\end{center}

All unlisted pairs are unacceptable. Thus every acceptable pair assigns the
members of a couple to distinct hospitals.

Write \(p_{jt}\) for the \(t\)th listed option of couple \(k_j\). A hospital
rejects doctor \(d\) if it is full and every current assignee precedes \(d\)
in the master list. Thus, when a couple does not receive a preferred option,
at least one of the two proposed doctors must be rejected at the corresponding
hospital. As every acceptable pair uses distinct hospitals, an S3 deviation
cannot occur.

Both singles must be assigned in any matching without an S1 block. If \(s_1\)
were unmatched, every hospital would have to be full with doctors preceding
\(s_1\), which would require nine positions to be occupied by only five such
doctors. The corresponding argument for \(s_2\) would require nine positions
to be occupied by the seven doctors preceding \(s_2\).

Assume, for a contradiction, that \(M\) has neither an S1 nor an S2 block.
We now apply the rejection condition successively to the preferred options of
the couples. At each step, an option that is not assigned must be rejected at
one of its two hospitals. The resulting cases are as follows.

Suppose first that \(q'_{h_1}=1\). At \((1,3,5)\), the rejection
conditions force \(s_1\) to \(h_2\), \(s_2\) to \(h_3\),
\(k_1\) to \(p_{12}\), \(k_2\) to \(p_{23}\), and \(k_3\) to \(p_{33}\),
leaving \(k_4\) unmatched. Couple \(k_4\) then blocks through \(p_{42}\):
\(d_7\) displaces \(d_3\) at \(h_2\), and \(d_8\) fills a vacant position at
\(h_3\). At \((1,4,4)\), the forced assignment sends \(s_1\) to \(h_2\),
\(s_2\) to \(h_3\), \(k_1\) to \(p_{11}\), \(k_2\) to \(p_{23}\), and
\(k_4\) to \(p_{42}\), leaving \(k_3\) unmatched. Couple \(k_3\) blocks
through \(p_{33}\), since \(d_5\) fills a vacancy at \(h_3\) and \(d_6\)
displaces \(d_1\) at \(h_2\). At \((1,5,3)\), the rejection conditions force
\(s_1,s_2\) to \(h_2\), \(k_1\) to \(p_{11}\), \(k_2\) to \(p_{23}\), and
\(k_3\) to \(p_{33}\), leaving \(k_4\) unmatched. Again \(k_4\) blocks
through \(p_{42}\): \(d_7\) displaces \(d_1\) at \(h_2\), and \(d_8\) fills
a vacancy at \(h_3\).

Let \(q'_{h_1}=2\). At \((2,2,5)\), the no-block conditions force
\(s_1\) to \(h_2\), \(s_2\) to \(h_3\), \(k_1\) to \(p_{12}\),
\(k_2\) to \(p_{23}\), and \(k_3\) to \(p_{31}\), leaving \(k_4\)
unmatched. Couple \(k_4\) blocks through \(p_{42}\): \(d_7\) displaces
\(d_3\) at \(h_2\), and \(d_8\) is accepted at a vacant position of \(h_3\).
At \((2,3,4)\), the forced assignment sends \(s_1\) to \(h_2\), \(s_2\) to
\(h_3\), \(k_1\) to \(p_{11}\), \(k_3\) to \(p_{31}\), and \(k_4\) to
\(p_{42}\), leaving \(k_2\) unmatched. Couple \(k_2\) then blocks through
\(p_{23}\), with \(d_3\) displacing \(d_1\) at \(h_2\) and \(d_4\) filling
a vacancy at \(h_3\). At \((2,4,3)\), the assignment is \(s_1,s_2\) to
\(h_2\), \(k_1\) to \(p_{11}\), \(k_2\) to \(p_{23}\), and \(k_3\) to
\(p_{31}\), leaving \(k_4\) unmatched. Couple \(k_4\) blocks through
\(p_{42}\), as \(d_7\) displaces \(d_1\) at \(h_2\) and \(d_8\) fills a
vacancy at \(h_3\). At \((2,5,2)\), the same assignment is forced, but now
\(d_7\) fills a vacancy at \(h_2\) and \(d_8\) displaces \(d_4\) at \(h_3\).
Thus \(p_{42}\) is again an S2 block for \(k_4\).

Let \(q'_{h_1}=3\). At \((3,1,5)\), the assignment is \(s_1\) to
\(h_2\), \(s_2\) to \(h_3\), \(k_1\) to \(p_{12}\), \(k_3\) to \(p_{31}\),
and \(k_4\) to \(p_{41}\), leaving \(k_2\) unmatched. Couple \(k_2\) blocks
through \(p_{23}\), because \(d_3\) displaces \(s_1\) at \(h_2\) and \(d_4\)
fills a vacancy at \(h_3\). At \((3,2,4)\), the corresponding assignment has
\(k_1\) at \(p_{11}\), with all other assignments unchanged. The same option
of \(k_2\) blocks: \(d_3\) displaces \(d_1\) at \(h_2\), while \(d_4\) is
accepted at \(h_3\). At \((3,3,3)\), the rejection conditions force \(s_1\)
to \(h_1\), \(s_2\) to \(h_2\), \(k_1\) to \(p_{11}\), \(k_2\) to
\(p_{23}\), and \(k_3\) to \(p_{31}\), leaving \(k_4\) unmatched. Couple
\(k_4\) blocks through \(p_{41}\), with \(d_7\) displacing \(s_1\) at \(h_1\)
and \(d_8\) filling a vacancy at \(h_3\). At \((3,4,2)\), the same assignment
is forced; now \(d_8\) displaces \(d_4\) at \(h_3\), so \(p_{41}\) still
blocks. At \((3,5,1)\), the forced assignment sends \(s_1\) to \(h_1\),
\(s_2\) to \(h_2\), \(k_1\) to \(p_{11}\), \(k_2\) to \(p_{24}\), and
\(k_3\) to \(p_{32}\), leaving \(k_4\) unmatched. Couple \(k_4\) blocks
through \(p_{43}\): \(d_7\) fills a vacancy at \(h_2\), and \(d_8\)
displaces \(s_1\) at \(h_1\).

Let \(q'_{h_1}=4\). At \((4,1,4)\), the rejection conditions force
\(s_1,s_2\) to \(h_1\), \(k_2\) to \(p_{23}\), \(k_3\) to \(p_{31}\), and
\(k_4\) to \(p_{41}\), leaving \(k_1\) unmatched. Couple \(k_1\) blocks
through \(p_{12}\), because \(d_1\) fills a vacancy at \(h_3\) and \(d_2\)
displaces \(d_6\) at \(h_1\). At \((4,2,3)\), the forced assignment sends
\(s_1\) to \(h_1\), \(s_2\) to \(h_2\), \(k_1\) to \(p_{11}\), \(k_3\) to
\(p_{31}\), and \(k_4\) to \(p_{41}\), leaving \(k_2\) unmatched. Couple
\(k_2\) blocks through \(p_{23}\): \(d_3\) displaces \(d_1\) at \(h_2\),
while \(d_4\) fills a vacancy at \(h_3\). At \((4,3,2)\), the same assignment
is forced; here \(d_3\) fills a vacancy at \(h_2\) and \(d_4\) displaces
\(d_5\) at \(h_3\), so \(p_{23}\) blocks. At \((4,4,1)\), the assignment
is \(s_1,s_2\) to \(h_1\), \(k_1\) to \(p_{11}\), \(k_2\) to \(p_{24}\), and
\(k_3\) to \(p_{32}\), leaving \(k_4\) unmatched. Couple \(k_4\) blocks
through \(p_{43}\), since \(d_7\) fills a vacancy at \(h_2\) and \(d_8\)
displaces \(s_1\) at \(h_1\).

Finally, let \(q'_{h_1}=5\). At \((5,1,3)\), the assignment is
\(s_1,s_2\) to \(h_1\), \(k_1\) to \(p_{11}\), \(k_3\) to \(p_{31}\), and
\(k_4\) to \(p_{41}\), leaving \(k_2\) unmatched. Couple \(k_2\) blocks
through \(p_{23}\): \(d_3\) displaces \(d_1\) at \(h_2\), and \(d_4\) fills
a vacancy at \(h_3\). At \((5,2,2)\), the forced assignment sends \(s_1,s_2\)
to \(h_1\), \(k_1\) to \(p_{11}\), \(k_2\) to \(p_{21}\), and \(k_4\) to
\(p_{41}\), leaving \(k_3\) unmatched. Couple \(k_3\) blocks through
\(p_{32}\), with \(d_5\) filling a vacancy at \(h_2\) and \(d_6\) displacing
\(d_4\) at \(h_1\). At \((5,3,1)\), the forced assignment sends \(s_1,s_2\)
to \(h_1\), \(k_1\) to \(p_{11}\), \(k_2\) to \(p_{21}\), and \(k_3\) to
\(p_{32}\), leaving \(k_4\) unmatched. Couple \(k_4\) blocks through
\(p_{43}\), because \(d_7\) fills a vacancy at \(h_2\) and \(d_8\)
displaces \(d_4\) at \(h_1\).

Every capacity vector satisfying the stated bounds appears in one of these
five cases. In every case, the presumed matching has an S2 block, a
contradiction. Hence no admissible capacity vector admits a stable matching.

\subsection{Proof of Theorem~\ref{thm:one-sided-near-feasibility}}
\label{app:one-sided-near-feasibility-proof}

We use the construction in the proof of
\citet[Theorem~3.1]{csaji2026}. Under the
sub-completeness assumed here, our requirement that both members be assigned
or both remain unmatched places every couple. The reduction
and rounding rule are unchanged.
We show that their output never reduces a hospital capacity, bound its total
expansion, and verify BIS stability.

\paragraph{The inherited construction.}
We recall the part of the reduction needed for our analysis. In the auxiliary
Stable Fixtures instance \(G\), each doctor \(i\in\D\) has capacity one and
each hospital \(h\in\Hh\) has capacity \(q_h\). Add an edge \(ih\) for every
acceptable doctor--hospital pair \((i,h)\in\E\). Hospitals and single doctors
retain their respective preference orders. For each couple \(k=(a,b)\), use
the individual orders \(\succ_a,\succ_b\) supplied in the input to establish
sub-responsiveness.

The connected couple gadget of \citet{csaji2026} adds four vertices
\(u_k^a,v_k^a,u_k^b,v_k^b\), each of capacity one, and the six edges
\[
au_k^a,\quad u_k^av_k^a,\quad v_k^ab,\quad
bu_k^b,\quad u_k^bv_k^b,\quad v_k^ba.
\]
Doctor \(a\) ranks \(u_k^a\) first, the hospitals in \(H(a)\) according to
\(\succ_a\) next, and \(v_k^b\) last. Doctor \(b\) similarly ranks
\(u_k^b\) first, her hospitals according to \(\succ_b\) next, and \(v_k^a\)
last. The connector preferences are
\[
v_k^a\succ_{u_k^a}a,\qquad b\succ_{v_k^a}u_k^a,\qquad
v_k^b\succ_{u_k^b}b,\qquad a\succ_{v_k^b}u_k^b.
\]

Let \(x\) be the stable half-integral matching used in their construction,
obtained by the Stable Fixtures procedure based on
\citet{tan1991,fleiner2010}. Thus \(x_e\in\{0,\tfrac12,1\}\), vertex
capacities are respected, and every edge with \(x_e<1\) is dominated at an
endpoint. More precisely, write \(\delta(v)\) for the edges incident to a
vertex \(v\) of capacity \(q_v\). An edge \(e\) is dominated at \(v\) if
\[
\sum_{\substack{f\in\delta(v)\\f\succeq_v e}}x_f=q_v.
\]
In particular, \(v\) is saturated and no edge below \(e\) has positive
weight there.

We also recall the orientation argument in
\citet[Observation~3.2]{csaji2026}. Call an edge with value \(\tfrac12\) a
half-edge and orient it towards an endpoint
where it is dominated. At most one half-edge enters any vertex: domination
of a better incoming edge would exclude positive weight on a worse one.
A vertex with an incoming half-edge is saturated and, because its capacity
is integral, has at least one further half-edge. This further edge must
leave the vertex. The outdegree is therefore at least the indegree at every
vertex. Equality of their global sums gives equality at each vertex.
Consequently, every vertex incident to a half-edge is saturated and has
exactly one incoming and one outgoing half-edge. The half-edges form
vertex-disjoint cycles.

The original rounding rule assigns each doctor \(i\) to her most preferred
hospital in
\(H_x(i)=\{h\in H(i):x_{ih}>0\}\), leaving her unmatched if this set is
empty. Denote the assignment by \(M\), and retain the original capacity rule
\[
\widehat q_h=
\begin{cases}
|M(h)|, & \text{if \(h\) is saturated in \(x\)},\\
q_h, & \text{otherwise}.
\end{cases}
\]

The feasibility argument in \citet[Claim~3.3]{csaji2026} ensures that this
rounding respects connected couples. We recall why. If a vertex of capacity
one ranks another such vertex first, the latter is saturated: their edge
either has value one or must be dominated at the latter endpoint. Each
vertex of a connected couple gadget is the first choice of another gadget
vertex, so all six are saturated. Suppose now that \(a\) has a positive
hospital edge but \(b\) does not, and put \(z=x_{bu_k^b}\). The saturation
equations give
\[
x_{v_k^ab}=1-z,\qquad x_{u_k^bv_k^b}=1-z,\qquad
x_{v_k^ba}=z,\qquad x_{u_k^av_k^a}=z,\qquad x_{au_k^a}=1-z.
\]
Hence \(x_{au_k^a}+x_{v_k^ba}=1\), leaving no capacity at \(a\) for a
positive hospital edge, a contradiction. The other case is symmetric.
Thus both members are assigned or neither is assigned. Sub-completeness
makes the assigned hospital pair acceptable.

\paragraph{Capacity bounds.}
We now strengthen the capacity guarantee for this same output. Every
doctor--hospital edge of value one is retained by the rounding rule. An unsaturated hospital
has no half-edge and retains its original capacity, so its assignment is
feasible. A saturated hospital with no half-edge receives exactly \(q_h\)
doctors and also has unchanged capacity.

Consider a hospital \(h\) incident to half-edges. It is saturated and has
exactly two such edges and \(q_h-1\) edges of value one. Let \(ih\) be the
half-edge entering \(h\). The other half-edge at \(i\) enters \(i\) and is
dominated there. If it joins \(i\) to another hospital \(g\), domination
implies \(h\succ_i g\), so the rounding selects \(h\). If it joins \(i\)
to a connector, it cannot be the first choice connector: an edge of value
one half cannot be dominated at a vertex of capacity one that ranks it
first. It is therefore the last choice connector, and \(h\) is the only
hospital with positive weight at \(i\). The rounding again selects \(h\).
Thus \(h\) retains at least one of its two fractional doctors, in addition
to all \(q_h-1\) doctors on edges of value one. It follows that
$q_h\leq |M(h)|\leq q_h+1$.
Since \(\widehat q_h=|M(h)|\), this proves
\(0\leq\widehat q_h-q_h\leq1\) for every hospital. In particular, the
capacity rule in the original construction makes no reductions on the
instances considered here.

To bound the total expansion, let \(n_1\) and \(n_2\) count doctor vertices
incident to exactly one and exactly two hospital half-edges, respectively.
The rounding retains \(n_1+n_2\) hospital half-edges. Each hospital incident
to half-edges has exactly two of them, so there are
\((n_1+2n_2)/2\) such hospitals. Their total load on half-edges is
\((n_1+2n_2)/2\), whereas the rounding retains \(n_1+n_2\) corresponding
assignments. All doctor--hospital edges of value one are unchanged. Therefore,
\[
\sum_{h\in\Hh}(\widehat q_h-q_h)
=(n_1+n_2)-\frac{n_1+2n_2}{2}
=\frac{n_1}{2}.
\]
A single doctor is adjacent only to hospitals. If incident to a half-edge,
she has two hospital half-edges and cannot contribute to \(n_1\). Hence
\(n_1\leq2|\C|\), giving the required bound
\(\sum_h(\widehat q_h-q_h)\leq|\C|\).

\paragraph{BIS stability.}
We make explicit why the same output is BIS-stable. For any acceptable
alternative preferred by a single doctor or a couple, we show that one
proposed doctor \(i\) and her target hospital \(r\) satisfy
\(x_{ir}=0\) and that \(ir\) is dominated at \(r\).

For a single doctor, a hospital \(r\) preferred to \(M(i)\) has
\(x_{ir}=0\) by the rounding rule. All positive edges at \(i\) lead to
less preferred hospitals, so \(ir\) is not dominated at \(i\) and must be
dominated at \(r\).

For a matched couple \(k=(a,b)\) that prefers \((h,g)\) to \(M(k)\), at
least one member strictly prefers her proposed hospital to her current
hospital. Otherwise, sub-completeness permits changing the two coordinates
one at a time, and sub-responsiveness gives \(M(k)\succeq_k(h,g)\), a
contradiction. For such a member \(i\) and target \(r\), the rounding rule
again gives \(x_{ir}=0\), with every positive hospital edge at \(i\) below
\(r\). The only connector ranked above \(r\) is the first choice connector.
Its weight is at most \(\tfrac12\), because \(i\) has a positive hospital
edge. Thus \(ir\) is not dominated at \(i\), so \(r\) dominates that edge.

For an unmatched couple \(k=(a,b)\), neither member has a positive hospital
edge. If both proposed edges \(ah\) and \(bg\) were dominated at their
doctor endpoints, then \(x_{au_k^a}=x_{bu_k^b}=1\), since only the first
choice connector is ranked above an acceptable hospital. These two edges
force both edges incident to \(v_k^a\) to have value zero. Then
\(u_k^av_k^a\) is dominated at neither endpoint: \(u_k^a\) prefers
\(v_k^a\) to \(a\), and \(v_k^a\) is unsaturated. This contradicts
stability of \(x\). Hence at least one proposed edge is dominated at its
hospital endpoint.

In every case, domination of \(ir\) at \(r\) makes \(r\) saturated in
\(x\), so \(\widehat q_r=|M(r)|\). Every assignee in \(M(r)\) has a
positive edge to \(r\). Because \(x_{ir}=0\) and hospital priorities are
strict, all these assignees rank above \(i\). Thus \(r\) is full, \(i\)
is not already assigned there, and no incumbent can be displaced by \(i\).
This excludes S1 and S2. For a proposed pair \((h,h)\), the same rejection
also excludes S3: its vacancy cases cannot hold, and every remaining case
requires both proposed doctors to rank above at least one incumbent.
Consequently, \(M\) is BIS-stable under \(\widehat q\).

\paragraph{Running time.}
The linear running time is inherited from the construction in
\citet[Theorem~3.1]{csaji2026}. The auxiliary graph has
\(|\D|+|\Hh|+4|\C|\) vertices and \(|\E|+6|\C|\) edges, and its ordered
adjacency lists have size \(O(L)\). The stable half-integral matching and
the rounding can therefore be computed in \(O(L)\) time. The orientation
and counting arguments above require at most a further linear scan and do
not change the algorithm. It uses the supplied individual orders, without
constructing or scanning the potentially quadratic lists of joint options.
As in the theorem's input convention, checking that an explicitly supplied
joint ranking satisfies the two preference restrictions is not part of the
construction. Every expanded hospital is full by the definition of
\(\widehat q\).

\subsection{Proof of Proposition~\ref{prop:tight-one-sided-budget}}
\label{app:tight-one-sided-budget-proof}

The universal upper bound \(B_{\min}(\mathcal I)\leq|\C|\) follows directly
from Theorem~\ref{thm:one-sided-near-feasibility}. It remains to show that this
bound is attained. 
For each \(k\in\{1,\ldots,m\}\), create three hospitals
\(h_k^1,h_k^2,h_k^3\) with capacities \((1,1,2)\), respectively. Component
\(k\) contains the couple \(k=(a_k,b_k)\) and two singles \(s_k,t_k\).
The couple accepts only \((h_k^1,h_k^2)\). The singles'
preferences are
\[
  s_k:h_k^1\succ h_k^2,
  \qquad
  t_k:h_k^3.
\]
The relevant hospital priorities are
\[
  h_k^1:a_k\succ s_k,
  \qquad
  h_k^2:s_k\succ b_k,
  \qquad
  h_k^3:t_k.
\]
There are no acceptable assignments between different components. A common
master list is obtained by placing all \(a_k\) first, then all \(s_k\), then
all \(b_k\), and finally all \(t_k\); the order within each group is
arbitrary. The couple list is both sub-complete and sub-responsive,
and its members are assigned to distinct hospitals. Each component has four
doctors and total capacity four.

Consider a component in which neither \(h_k^1\) nor \(h_k^2\) is expanded.
If the couple is assigned, \(s_k\) is unmatched and blocks with \(h_k^2\),
which prefers \(s_k\) to \(b_k\). If the couple is unmatched, an unmatched
\(s_k\) blocks with an unfilled hospital; assigning \(s_k\) to \(h_k^1\)
creates the couple block \((h_k^1,h_k^2)\), while assigning \(s_k\) to
\(h_k^2\) creates the single-doctor block \((s_k,h_k^1)\). Thus a stable
matching requires an expansion at \(h_k^1\) or \(h_k^2\) in every component.
Because components have disjoint acceptable assignments, any stable solution
uses at least \(m\) added positions.

Conversely, increase the capacity of every \(h_k^1\) from one to two. Assign
\(a_k,s_k\) to \(h_k^1\), \(b_k\) to \(h_k^2\), and \(t_k\) to \(h_k^3\).
Every doctor receives a most preferred acceptable assignment and every couple
receives its only acceptable pair, so the union of these component matchings
is stable. Hence \(B_{\min}(\mathcal I)=m\), proving that the universal bound
is tight.

\subsection{Proof of Theorem~\ref{thm:formulation-exactness}}
\label{app:formulation-exactness-proof}

Consider any feasible integer solution of ICM
\eqref{eq:variable-model}. The assignment and linking constraints define a
matching \(M\), and \eqref{eq:general-capacity-change} defines a nonnegative
integer capacity vector \(q'\). Conditional on this vector, the matching,
capacity, and stability constraints are precisely those of IFC
\eqref{eq:fixed-model} evaluated at \(q'\). The fixed capacity equivalence
in Section~\ref{subsec:fixed-ip} therefore implies that \(M\) is BIS-stable
under \(q'\).

Conversely, let \(M\) be a BIS-stable matching under a nonnegative integer
capacity vector \(q'\). Set \(x\) and \(y\) to the indicators of the
assignments in \(M\), and define
\[
  u_h^+=\max\{q'_h-q_h,0\},
  \qquad
  u_h^-=\max\{q_h-q'_h,0\}
  \quad\forall h\in\Hh.
\]
These variables satisfy \eqref{eq:general-capacity-change}. By the same
fixed capacity equivalence, the assignment variables can be extended with
binary S2 auxiliary variables \(\alpha\) to satisfy IFC at \(q'\).
Together they give a feasible solution of ICM.

Restricting \(q'\) to
\eqref{eq:balanced-capacity-region},
\eqref{eq:unit-expansion-capacity-region}, or
\eqref{eq:near-balanced-capacity-region} gives the corresponding equivalence
for \ICMBC{}, \ICMUC{}, or \ICMNF{}, respectively. Under balanced
modification, the construction above satisfies \(u_h^++u_h^-\leq1\) and
\(\sum_hu_h^+=\sum_hu_h^-\). Under unit expansion, it satisfies
\(u_h^-=0\), \(u_h^+\in\{0,1\}\), and \(\sum_hu_h^+\leq B\).
Under near feasible modification, the signed changes lie between \(-2\)
and \(2\), and their sum lies between \(0\) and \(1\). Conversely, these
restrictions imply their respective capacity regions.

The bounds \(U_h=q_h+1\) for \ICMBC{} and \ICMUC{}, and \(U_h=q_h+2\)
for \ICMNF{}, are valid throughout these regions. Hence the linearization
in Section~\ref{subsec:variable-ip} preserves the integer feasible solutions
in each case. Finally, \(\sum_{i\in\D}\sum_{h\in H(i)}x_{ih}\) counts the
matched doctors, so maximizing it selects a largest stable matching over the
admissible capacity vectors. Changing the objective to a capacity measure
does not alter this feasibility correspondence.

\subsection{Proof of Theorem~\ref{thm:cutoff-equivalence}}
\label{app:cutoff-equivalence-proof}

Consider first a feasible integer solution of the cutoff formulation.
Suppose a single doctor \(i\) has not obtained \(h\) or a better hospital.
The single-doctor prefix recursion gives \(b_{i,\rho_i(h)}=0\), so
\eqref{eq:cutoff-S1} gives \(t_h\geq s_{ih}+1>0\). Constraint
\eqref{eq:cutoff-fullness} then implies that \(h\) is full. Moreover,
\eqref{eq:cutoff-assignee} gives \(s_{jh}\geq t_h>s_{ih}\) for every doctor
\(j\) assigned to \(h\). Thus \(h\) ranks every current assignee above
\(i\), and \((i,h)\) is not an S1 block.

Next consider a couple that prefers \(p=p_k^t=(h,g)\), where \(h\neq g\),
to its current assignment. The couple-prefix recursion gives \(a_{kt}=0\).
Constraint~\eqref{eq:fixed-S2c} requires at least one of
\(\alpha_{kp}^{1}\) and \(\alpha_{kp}^{2}\) to be zero. Suppose, without
loss of generality, that \(\alpha_{kp}^{1}=0\). Constraint
\eqref{eq:cutoff-S2a} gives \(t_h>s_{d_k^1h}\). The same argument as for S1
shows that \(h\) is full and ranks every assignee above \(d_k^1\). Hence
\(h\) cannot accommodate the proposed member, so the couple does not form an
S2 block.

Finally, consider an option \(p_k^t=(h,h)\) that the couple prefers to its
assignment, so \(a_{kt}=0\), and relabel its members so that
\(i_1\succ_h i_2\). The hospital-tail recursion gives
$w_{h,r_h(i_2)-1}+x_{i_1h}=F_h+L_h(i_2)+x_{i_1h}+x_{i_2h}$,
where \(F_h=q'_h-|M(h)|\) and
\(L_h(i_2)=|\{j\in M(h):i_2\succ_h j\}|\). Because \(a_{kt}=0\), the two
proposed members are not both assigned to \(h\). The right-hand side is at
least two if and only if one of conditions~\textnormal{(S3)(i)},
\textnormal{(S3)(ii)}, \textnormal{(S3)(iii)(a)}, or
\textnormal{(S3)(iii)(c)} in Definition~\ref{def:bis-blocking} holds.
Constraint~\eqref{eq:cutoff-S3-base} therefore excludes precisely these
cases. If a couple
\(\ell\) is assigned to \((h,h)\), constraint
\eqref{eq:linked-cutoff-assignee} gives
\(t_h^L\leq\underline s_{\ell h}\), whereas
\eqref{eq:linked-cutoff-candidate} gives
\(s_{i_2h}+1\leq t_h^L\). Thus \(h\) cannot prefer both proposed members to
the lower-priority member of that assigned couple, excluding
condition~\textnormal{(S3)(iii)(b)}. The induced matching is
stable.

Conversely, fix a stable matching under a capacity vector. Set \(b\) and \(a\) using the displayed single-doctor and couple
prefix recursions. If hospital \(h\) has an unfilled position, set
\(z_h=t_h=0\). If it is full and \(q'_h>0\), set \(z_h=1\) and let \(t_h\)
be the lowest score among its assignees. If \(q'_h=0\), set \(z_h=1\) and
\(t_h=M_s\). These choices satisfy \eqref{eq:cutoff-assignee} and
\eqref{eq:cutoff-fullness}.

If \(b_{i,\rho_i(h)}=1\), constraint~\eqref{eq:cutoff-S1} is automatic. If
it is zero, the absence of an S1 block implies that \(h\) is full and every
assignee has a score of at least \(s_{ih}+1\); hence the chosen cutoff
satisfies the constraint. For an option \(p=p_k^t=(h,g)\) with \(h\neq g\),
set both \(\alpha\) variables to zero when the couple has obtained \(p\) or a
better option. In that case \(a_{kt}=1\), so both cutoff inequalities are
automatic. Otherwise \(a_{kt}=0\), and the absence of an S2 block means that
at least one proposed hospital cannot accommodate its member. Set the
corresponding \(\alpha\) variable to zero. Set the other to one if its cutoff
inequality would otherwise fail, and to zero otherwise. Then
\eqref{eq:cutoff-S2a}--\eqref{eq:cutoff-S2b} and
\eqref{eq:fixed-S2c} hold.

Set \(w\) using the hospital-tail recursion. If \(a_{kt}=0\), the absence of
conditions~\textnormal{(S3)(i)},
\textnormal{(S3)(ii)}, \textnormal{(S3)(iii)(a)}, and
\textnormal{(S3)(iii)(c)} gives \eqref{eq:cutoff-S3-base}. When
\(a_{kt}=1\),
\[
  w_{h,r_h(i_2)-1}+x_{i_1h}
  \leq q'_h\leq U_h=1+(U_h-1)a_{kt},
\]
so the same constraint is automatic. Set \(t_h^L\) to the smallest
lower-member score among the couples assigned to \((h,h)\), or to \(M_s\)
if there is none. Constraint~\eqref{eq:linked-cutoff-assignee} holds by
construction. When \(a_{kt}=0\), the preceding inequality together with the
absence of condition~\textnormal{(S3)(iii)(b)} implies that \(i_2\) is
ranked below the lower-priority member of every couple assigned to
\((h,h)\). If there is no such couple,
\(s_{i_2h}+1\leq M_s=t_h^L\). Otherwise, the consecutive integer scores
give \(s_{i_2h}+1\leq t_h^L\). Hence
\eqref{eq:linked-cutoff-candidate} holds; when \(a_{kt}=1\), that constraint
is automatic. Thus the stable matching extends to the cutoff formulation.
Theorem
\ref{thm:formulation-exactness} gives the same set of matchings and capacity
vectors for the original formulation.

\begingroup
\footnotesize
\setlength{\bibsep}{2pt}
\bibliographystyle{cas-model2-names-nodoi}
\bibliography{references}
\endgroup

\end{document}